\documentclass[aps,prx,twocolumn,superscriptaddress,export,nofootinbib]{revtex4-1}
\usepackage{hyperref}
\hypersetup{
colorlinks = true,
linkcolor=blue,
citecolor=[rgb]{0.15,0.65,0.13},
urlcolor = [rgb]{0.25, 0.41, 0.88}}
\usepackage{comment}

\usepackage{amsthm}
\newtheorem{theorem}{Theorem}

\usepackage[utf8]{inputenc}
\usepackage[german, english]{babel}
\usepackage[dvipsnames]{xcolor}
\usepackage{amsmath,amsfonts, amssymb, amsthm, dsfont}
\usepackage{bm}
\usepackage{graphicx}
\usepackage{tikz}
\usepackage{physics}
\usepackage{array, makecell} %
\usepackage[normalem]{ulem}
\usepackage{xspace}
\usepackage[export]{adjustbox}
\usepackage{subfigure}
\usepackage{soul}

\usepackage{enumerate}

\newcommand{\RN}[1]{%
  \textup{\uppercase\expandafter{\romannumeral#1}}%
}

\newcommand{\SC}{\text{SC}}

\begin{document}
\newcommand{\condition}{\text{symmetry-decoupling }}
\newcommand{\Condition}{\text{Symmetry-decoupling }}
\title{Quantum Communication and Mixed-State Order\\ in Decohered Symmetry-Protected Topological States}
\author{Zhehao Zhang}
\affiliation{Department of Physics, University of California, Santa Barbara, CA 93106, USA}	

\author{Utkarsh Agrawal}
\affiliation{Kavli Institute for Theoretical Physics, Santa Barbara, CA 93106, USA}
				
\author{Sagar Vijay}
\affiliation{Department of Physics, University of California, Santa Barbara, CA 93106, USA}	

\begin{abstract}
Certain pure-state symmetry-protected topological orders (SPT) can be used as a resource for transmitting quantum information. Here, we investigate the ability to transmit quantum information using decohered SPT states, and relate this property to the ``strange correlation functions" which diagnose quantum many-body orders in these mixed-states.  This perspective leads to the identification of a class of quantum channels -- termed symmetry-decoupling channels -- which do not necessarily preserve any weak or strong symmetries of the SPT state, but nevertheless protect quantum many-body order in the decohered mixed-state. We quantify the ability to transmit quantum information in decohered SPT states through the coherent quantum information, whose behavior is generally related to a decoding problem, whereby local measurements in the system are used to attempt to ``learn" the symmetry charge of the SPT state before decoherence.   
\end{abstract}

 \maketitle

\tableofcontents

\vspace{2em}

\section{Introduction}
The interplay between symmetries and quantum entanglement in a quantum many-body system can give rise to symmetry-protected topological (SPT) orders \cite{senthil2015symmetry}: gapped, symmetric states of zero-temperature quantum matter which are nevertheless distinct when the symmetry is preserved. Recently, the behavior of SPT orders under the effects of interactions with an external environment have been considered \cite{de2022symmetry,chen2023symmetry, lee2023symmetry, zhang2022strange, ma2023average, ma2023topological,  roberts2017symmetry, roberts20243, lessa2024mixed, xue2024tensor, guo2024locally, ma2024symmetry, Wang_li_2024_double_state, chirame2024stable}, partly motivated by the advent of increasingly powerful quantum hardware capable of imperfectly engineering quantum states \cite{chen2023realizing, foss2023experimental, bluvstein2024logical, iqbal2024non, zhu2023nishimori, lee2022decoding,lu2023mixed, tantivasadakarn2021long, tantivasadakarn2023hierarchy, lu2022measurement}  in the presence of spurious interactions with an external environment. Understanding the robustness of quantum collective phenomena in this setting is an important fundamental question, with additional significance if the state of interest is to be used for quantum applications.

\begin{figure*}
    \centering
    \includegraphics[width=1\textwidth]{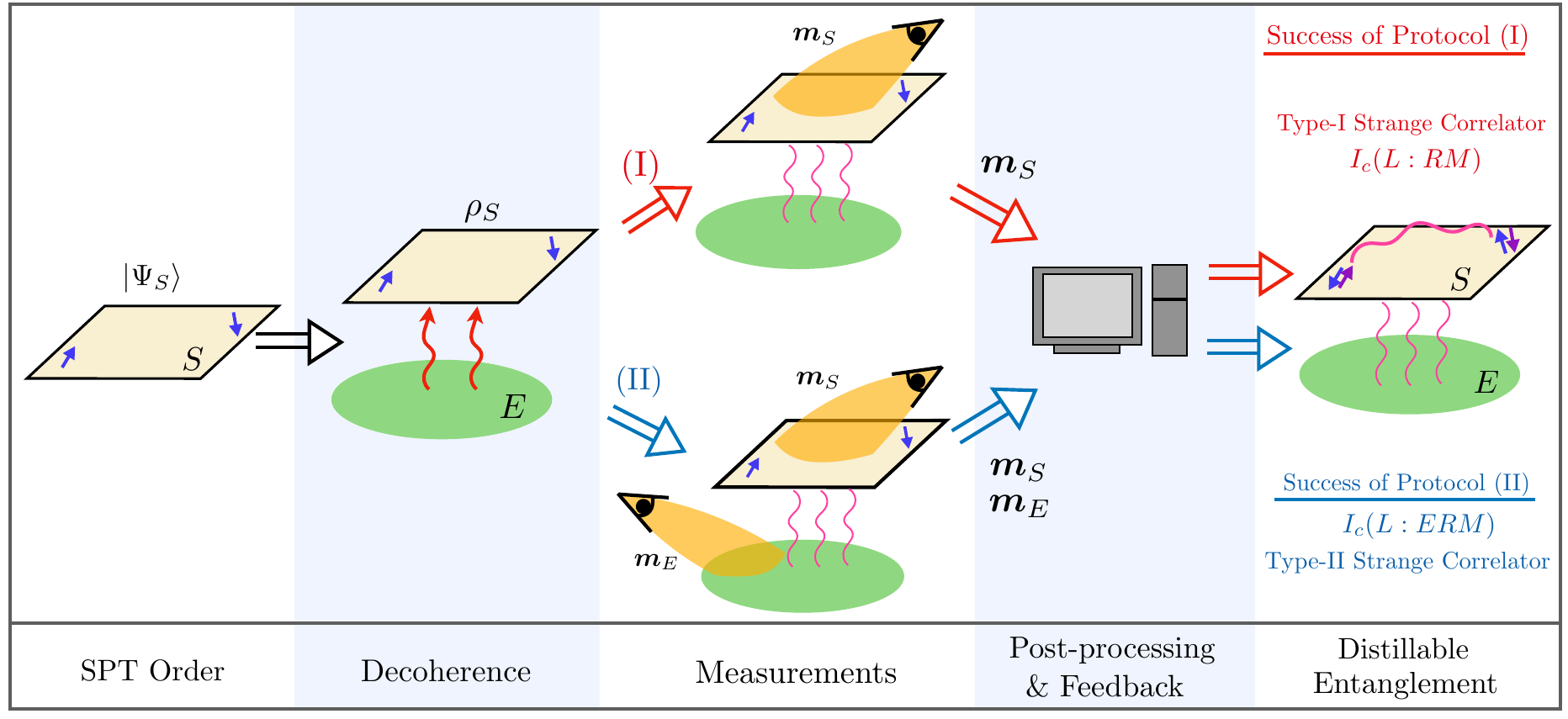}
    \caption{Schematic depiction of the setup considered in this work. The system of interest $S$ is prepared in an SPT state, and ancillas (not shown) are entangled with one end of the the system. After bulk decoherence, we determine whether quantum information can be transferred across the system, by measuring ($I$) the symmetry charge in the bulk of the SPT order, and ($II$) with possible additional measurements of the environment $E$.   Feedback, based on these measurement outcomes is used to distill quantum entanglement. The success of either of these protocols only depends on the decoherence channel that is applied, and is insensitive to unphysical properties of the purification of this channel into a particular state of the environment.  The ability to distill entanglement in these settings is quantified by the coherent quantum information and is directly related to the ``type-I" and ``type-II" strange correlation functions in the decohered state, as indicated.}
    \label{fig: Fig 1}
\end{figure*}

Investigating the behavior of SPT states under the effects of local decoherence is particularly important in this context, since the ground-states for certain SPT orders are believed to be resource states for measurement-based quantum computation (MBQC) and teleportation~\cite{nautrup2015symmetry,wei2017universal,miller2015resource,stephen2017computational,raussendorf2017symmetry,raussendorf2019computationally,devakul2018universal, raussendorf2023measurement,else2012symmetry,Raussendorf_2001,Raussendorf_2003,Briegel_2009}; measurements on these SPT wavefunctions, followed by unitary feedback, can be used to distill patterns of long-range entanglement.  This property is intimately related to universal aspects of these specific SPT orders, which possess a mixed anomaly between two on-site, unitary symmetries. This property ensures that when one symmetry charge is measured within a contiguous region, the degrees of freedom at the boundaries become entangled~\cite{Marvian_2017,lu2023mixed,tantivasadakarn2021long}.  Non-local observables (string/membrane order parameters) in the SPT ground-state~\cite{pollmann2012detection} also diagnose the interplay between these symmetries~\cite{Else_2014_SPT_classification}, while strange correlators~\cite{you2014wave} detect the non-trivial interface between these SPT orders and trivial, gapped symmetric pure-states.  These probes of quantum many-body order further reveal that a pure SPT state cannot be prepared from a symmetric, product state by a constant-depth, symmetric unitary circuit~\cite{chen2010local,huang2015quantum}.    

In this work, we study the ability of \emph{decohered} SPT states to transmit quantum information, and  connect this property to the orders within these mixed quantum many-body states.  Specifically, we find that the ability to use decohered SPT states as a resource for quantum information transfer is related to strange correlation functions \cite{you2014wave,lee2023symmetry,zhang2022strange}, which are sensitive to orders in these mixed-states, when they are interpreted as wavefunctions in a doubled Hilbert space. More precisely, quantum information which can be transferred by performing symmetric measurements in the bulk of the decohered SPT is related to long-range-order in the ``type-I" strange correlation function, while successful information transfer which can be achieved with the additional feature that one has access to the decohering environment is related to order in ``type-II" strange-correlators \cite{lee2023symmetry}.  The former property indicates that the mixed SPT state behaves like a quantum error-correcting code~\cite{dennis2002topological, raussendorf2007fault}, which can have a threshold strength of decoherence, below which coherent information transfer is possible~\cite{fan2023diagnostics, su2024tapestry,li2024replica, chen2023separability, sang2023mixed, myerson2023decoherence, chen2024unconventional, sohal2024noisy, sang2024stability, lu2024disentangling, ellison2024classification,lavasani2024stability}. These protocols for transmitting quantum information in decohered SPT states, and the connection to mixed-state many-body orders, are shown schematically in Fig. \ref{fig: Fig 1}.

The relationship between the ability to use decohered SPT states as a resource for quantum communication, and mixed-state order is particularly fruitful, as it allows us to introduce a class of decoherence channels -- termed \emph{\condition} channels -- which do not preserve any (strong or weak) symmetries of the mixed SPT state, but which nevertheless preserve orders in the resulting decohered density matrix, as diagnosed by non-trivial strange correlation functions.  An understanding of this result by considering the density matrix as the ground-state of a quantum many-body system in a doubled Hilbert space -- a setting in which symmetries would play a privileged role in interpreting the resulting mixed-state orders -- remains an open question of our work.  We also argue that that the ability to perform quantum communication using mixed SPTs implies that such mixed states cannot be prepared in finite time and resources~\cite{chen2023symmetry}, suggesting that mixed SPT states obtained through the action of a \condition channel could be in a different mixed quantum many-body phase from trivial symmetric states which have been affected by decoherence.  We leave thorough investigation of this possibility to future work.  

We introduce two probes of coherent quantum information that quantify the amount of  information that can be transferred through a decohered SPT state by performing symmetric bulk measurements, and perform calculations of these for SPTs in various dimensions, where these information-theoretic quantities exhibit phase transitions as the strength of the bulk decoherence is increased.   
The coherent quantum information is advantageous as a diagnostic of mixed-SPT order. The behavior of strange correlation functions can be sensitive to the choice of trivial, symmetric state which defines these correlations, while no such ambiguity is encountered in the definition of the coherent quantum information.  Furthermore, in examples where quantum information can be successfully transmitted across boundaries of a decohered SPT, the coherent quantum information makes precise the notion of a quantum-coherent ``edge" of the mixed SPT state. 

\subsection{Setup and Summary of Results}

We now provide a detailed summary of our results. The general setup that we study begins with a system in a symmetric trivial state, with respect to an on-site unitary symmetry. The system has open (left and right) boundaries. We appropriately entangle ancilla qubits $L$ to the left boundary of the system, and then prepare the remaining qubits in an SPT state with respect to this symmetry.  The precise nature of the entanglement of the left boundary with the ancilla may depend on the SPT state, as we clarify in subsequent sections.     
We then perform measurements of the on-site symmetry charge at all sites, except at the two boundaries, in order to attempt to coherently transmit quantum information between these regions~\cite{Marvian_2017}.   The general setup  for a pure-state SPT order in one spatial dimension (without any decoherence) is indicated Fig.~\ref{fig: MBQC setup}.

For pure SPTs without decoherence, some amount of quantum information in the ancilla qubits $L$  gets transmitted to the right boundary $R$, which can be ``decoded" using the bulk measurement outcomes.  This property can also provide a probe of the SPT order. {In the SPT phase this ability to transmit quantum information is quantified by the coherent quantum information $I^{(0)}_c(L:RM)>0$ which is zero in a trivial symmetric state (the superscript $(0)$ is to denote the coherent quantum information when the \emph{pure} SPT state is used to transmit information).} Here $M$ denotes the classical set of measurement outcomes in the bulk. 
The theory of approximate quantum error-correction \cite{Schumacher_1996} guarantees that $I_c^{(0)}(L:RM)$ provides a lower bound to the fidelity with which the state of $L$ can be transmitted to $R$, $M$ via this protocol.

Due to the non-zero information being transmitted, there are logical spaces at the $L$ and $R$ boundaries that are entangled after measurement of the symmetry charge. The decoding or recovery of the information then involves distilling the logical space from the boundary space and identifying the particular nature of the entanglement between them.

An appropriate recovery map is then applied at the right boundary to recover the logical information. We review these concepts for pure cluster state SPT orders in various dimensions in Sec.~\ref{sec: MBQC pure}.

\begin{figure}
    \centering
    \includegraphics[width=\linewidth]{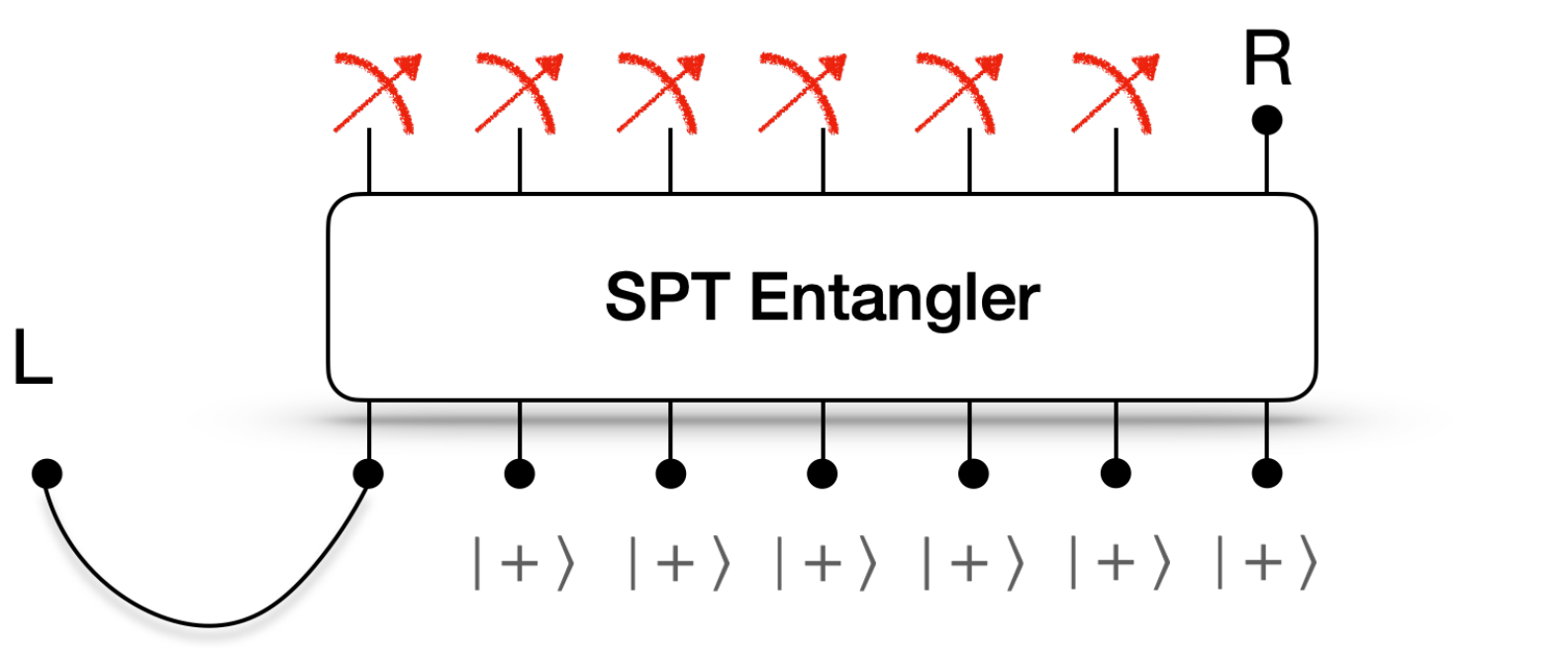}
    \caption{Transmitting quantum information using a one-dimensional, pure symmetry-protected topological order.  An ancilla ($L$) is entangled with the left boundary of a symmetry-protected topological order.  Bulk measurements of the symmetry charge can generate entanglement between $L$ and $R$, which can be used to send quantum information.}
    \label{fig: MBQC setup}
\end{figure}

A similar setup can be considered for a decohered SPT state, where we now measure the on-site symmetry charge in the bulk of the mixed SPT state. In standard quantum error correction, where logical information may be retained or leaked into the environment, access to the environment qubits allows for perfect recovery of the logical information. However, measurements in the symmetry  basis in the decohered SPT can in principle \textit{destroy} information about the symmetry charge. Therefore, quantum information may not be perfectly transmitted in this setting, even if one has access to the environment decohering the SPT state.

To quantify the ability to transmit information in this setting, we therefore introduce two types of coherent information, 1) information transmission with access to the environment, $I_c(L:ERM)$, and  2) information transmission without having access to the environment, $I_c(L:RM)$. Our first result identifies the classes of decoherence channels that allow information to be perfectly transmitted after measurements of the symmetry charge in the system, and with additional access to the environment. This class of channels, which we refer to as \condition channels, contains as a subset, weakly-symmetric channels (see their definitions in Sec.~\ref{sec: Information transmission using Mixed SPT}) which have been previously studied~\cite{lee2023symmetry, ma2023average, de2022symmetry}. Interestingly there are classes of \condition channels that are \textit{not} weakly symmetric, but allow for perfect transmission of quantum information. 

A more practical scenario is when we don't have access to the environment. Whether we can still use the system alone as a resource for information transmission is not obvious. $I_c(L:RM)$ quantifies the information recoverable about $L$ using only the right boundary $R$ and measurement outcomes on the system. We prove a bound on this quantity which depends on the amount of information about the environment that can be decoded using the system's measurement outcomes $M$. More precisely, if we can learn the value of the symmetry charge ``leaked'' into the environment from the measurement outcomes on the system then $I_c(L:RM)=I_c(L:ERM)$. That is, even in the presence of decoherence, the SPT state can be used as a resource state for quantum information transmission with the same fidelity as without decoherence. 

We also show that $I_c(L:RM)$ serves as an order parameter to distinguish different mixed-state phases based on their power to transmit quantum information. We want to emphasize that $I_c(L:RM)$ probes the intrinsic property of a mixed SPT state, and is independent of the existence of any environment state. In the example of $Z$-dephased 2d cluster state, we identify two phases distinguished by their ability to transmit quantum information, as is diagnosed by $I_c(L:RM)$. The phase diagram is the same as the previous work~\cite{chen2023symmetry}, which uses an Edwards-Anderson-like string-order parameter.

One of the key results of the paper is that the above coherent information quantities are directly related to the two types of strange correlators introduced in~\cite{lee2023symmetry}. More precisely, $I_c(L:RM)$ and $I_c(L:ERM)$ being non-zero implies that type-I and type-II strange correlators with {respect to a typical trivial density matrix} are, respectively, long-ranged.

\subsection{Symmetry-Decoupling Channels}
The \condition channels provide a new result of the paper. These channels are not necessarily strongly nor weakly symmetric. We first recall that in a strongly-symmetric channel the symmetry charge of the system remains entirely in the system, while for a weakly-symmetric channel, the charge is partially transferred into the environment and the total charge is the sum of charge in the system and environment. 
On the other hand, the charge in the \condition channels has no well-defined charge operator supported on the system and environment independently and the total charge is not given by a sum of charges from the two. However, after measurement of the original charge within the system, the remainder of the system becomes weakly symmetric, that is, the symmetry charge is now a sum of the measured charge on the system and a corresponding charge operator on the environment. 

We illustrate these ideas with a simple example. {Consider a system with an on-site $Z_{2}$ symmetry, and $X$ measures the on-site symmetry charge at a single qubit.}  After interactions with an external environment at that site, this operator evolves into $X'$. For weakly-symmetric channels, the modified charge is $X'=X\otimes O$ for some operator $O$ supported within the environment, and the total charge is the product of charges associated with $X$ and $O$. In contrast, the general action of a \condition channel yields $X'= X\otimes O_1 + I\otimes O_2$ (up to overall normalization) where $O_{1,2}$ are operators which are exclusively supported on the environment. After measurement of $X$ in the system, yielding outcome $m=\pm 1$, the symmetry charge is given by $X' = mO_1 + O_2$ which is only supported in the environment. {Thus the symmetry charge ``decouples" from the system and has perfectly leaked in to the environment after measurements within the system are performed. The  post-measurement state of the system is thus weakly symmetric. This suggests an alternate definition for symmetry-decoupling channels:  these are the channels for which strong on-site projective measurements of the symmetry charge a sub-region $A$, after application of this channel, leads to the emergence of a weak symmetry in the density matrix for the complementary region, 
as depicted in Fig.~\ref{fig: SDC intuition}.}

\begin{figure}
    \centering
    \includegraphics[width=1\linewidth]{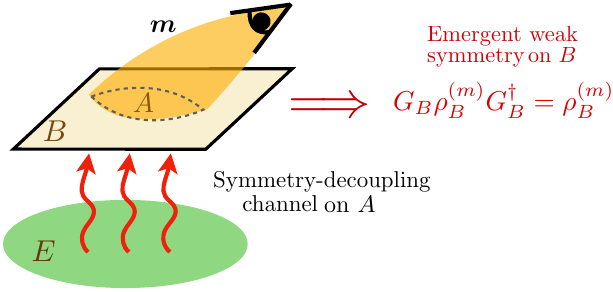}
    \caption{Schematic depiction of a symmetry-decoupling channel affecting region $A$. Subsequently, the symmetry charge in $A$ is measured.  The resulting density matrix in the complement $B$ , $\rho_B^{(m)}$, now has an emergent weak symmetry. For generic non-symmetry-decoupling channels, the post-measurement state $\rho_B^{(m)}$ is not weakly symmetric.}
    \label{fig: SDC intuition}
\end{figure}

We also present a quantum error-correction perspective on these probes. More precisely, the measurement of the bulk qubits can be interpreted as a unitary quantum evolution in a virtual time direction. Decoherence introduces noise in this virtual evolution. The behavior of coherent information $I_c(L:RM)$ is thus mapped to the decodability of the noisy quantum evolution. If the virtual quantum dynamics has an error threshold then the coherent information $I_c(L:RM)$ is positive for decoherence strengths below this threshold. 

The above results are believed to hold for any initial pure-state which lies within the SPT phase. Assuming that information transition in the pure case is stable against symmetric perturbations, a symmetric low-depth circuit on the system before decoherence should not change the above results as long as we thicken the boundaries with the thickness scaling with the depth of the circuit. This implies that the SPT phase consisting of ground states of SPT Hamiltonians is stable against decoherence since any pure state from SPT phase can still be used for quantum communication after being acted upon by decoherence.

\section{Information Transmission with pure SPT Order}\label{sec: MBQC pure}

Before moving to the discussion of mixed states we briefly review the protocol for pure states. This also helps in setting up the notation for later use.
As noted in the previous section, the SPT phase can be characterized by its resourcefulness to transmit quantum information from one end to another by measurement of bulk qubits in a symmetric basis. The setup is illustrated in Fig.~\ref{fig: MBQC setup}. We entangle $d_L$ ancillas $L$ with the first layer of qubits of the system at $i=1$, where $d_L$ is a number which we will clarify down below. We measure the system qubits except for qubits at the right-most boundary $R$. Let the measurement outcome be $m$. The amount of coherent information in $L$ transmitted to $R$ is given by\begin{align} \label{eq: type-1 coherent information from mutual information}
    I_c(L:RM) = \sum_m p_m (I^{(m)}(L:R)-S(\rho_L^{(m)})),
\end{align}
where $p_m$ is the probability of observing measurement outcome $m$, $\rho^{(m)}$ is the density matrix on the boundaries post-measurement, $I^{(m)}(L:R)$ is the mutual information between $L$ and $R$ for given measurement outcome $m$. See Appendix~\ref{app: i_c} for proof of the above expression. Since for pure SPTs $L,R$ are in a pure state after measurement the coherent information can be written as \begin{align}
I_c(L:RM) = \sum_m p_m S(\rho_R^{(m)}). \label{eq: I_c pure SPT}    
\end{align}
The state is in the SPT phase if the coherent information is non-zero and positive in the thermodynamic limit ~$N\rightarrow \infty$.
 In other words, the boundaries become long-range entangled post-measurement. 
 
\subsection{Logical operators} \label{sec: logical space}

For positive coherent information in eq.~\eqref{eq: I_c pure SPT} there are logical operators on the $L$ boundary that are transmitted to the $R$ boundary. We identify these logical operators as the restriction of symmetry action on the boundary of the systems as follows. 

Let $G=\prod_i g_i$ be the symmetry action of an element of the symmetry group, where $i$ runs over all qubits and $g_i$ is the on-site action of the symmetry. 
Starting from an eigenstate of $G$, post measurement of $g_i$ in the bulk the state is an eigenstate of $\prod_{i\in L}g_i \prod_{j\in R}g_j$, where $i\in L,R$ are qubits on the left, right boundary. $R$ can measure $\prod_{i\in R}g_i$ and learn the value of symmetric charge in $L$. This immediately implies the transmission of the classical information stored in the operator $G_L=\prod_{i\in L}g_i$ as follows: initialize the $L$ boundary in eigenstates of $G_L$ according to some classical probability distribution. The eigenvalue of $G_L$ can be determined by measuring the qubits in the bulk and the $R$ boundary. Repeated run thus allows $R$, with the help of bulk outcomes, to learn the probability distribution (or the classical information). Furthermore, quantum information is transmitted when the restrictions for different group elements are not commuting. This suggests a connection between mixed anomalies between different symmetry groups/elements and the transmission of quantum information.

In 1d SPT states with open boundaries, the symmetry action when restricted to the boundary forms a projective representation of the symmetry group. These projective representations consist of non-commuting operators and thus result in transmitting quantum information. For projective representation $\omega$ of the group $G$, if no two group elements commute then the number of qubits in the logical space is given by $d_L=\log_2\sqrt{|G|}$, where $|G|$ is the size of the group $G$ (recall we are working with finite Abelian groups); see Chapter 6, Theorem 6.6 in~\cite{Group_2019}. This condition is often referred to as maximally noncommutative (MNC) condition~\cite{else2012symmetry}. The size of the logical space is reduced when the representation $\omega$ contains commuting operators~\cite{Group_2019}.

As argued above, a logical space exists on the boundaries of SPTs. This logical space can be decoupled from the non-logical space by applying a unitary localized at the boundaries. That is, after measuring all the qubits in, the bulk the density matrix of the boundaries plus the measurement outcomes can be written as,
\begin{align}
    &\rho_{L,R,M} = \sum_m p_m \rho_{L,R}^{(m)}\ket{m}\bra{m} \label{eq: logical space decomposition} \\
    &= U^\dagger  \sum_m p_m \ket{\Phi^m}\bra{\Phi^m}\otimes \rho_{L,\mathrm{rest}}^m \rho_{R,\mathrm{rest}}^m \otimes \ket{m}\bra{m} U,\nonumber
\end{align}
where $U=U_L U_R$ are local unitary rotation on $L,R$ to decouple the logical space, $\ket{\Phi^{m}}$ is an entangled state between $L,R$ residing in the logical space, and $\rho^{(m)}_{L,R,\mathrm{rest}}$ are the remaining non-logical degrees of freedom on the boundaries not carrying any logical information. We will skip writing the unitary rotation $U$ for brevity. It can be easily seen that ~{$I_c(L:RM) = \sum_m p_m S(\mathrm{Tr}_L \ket{\Phi^m}\bra{\Phi^m})$}. $\ket{\Phi^m}$ belongs to a $2^{d_L}$ dimensional Hilbert space where $d_L$ is the number of logical qubits.
Different measurement outcomes $m$ can lead to the same $\ket{\Phi^m}$ and thus we define equivalency class $\gamma$ as values of $m$ resulting in the same $\ket{\Phi^m}$. $\gamma$ can be thought of as a many-to-one function from the space of measurement outcomes $m$ to the logical space. We therefore denote the state on the logical space by $\ket{\Phi_\gamma}$. Note that the states $\ket{\Phi_\gamma}$ need not be orthogonal. We then have, \begin{align}
    \rho_{L,R,M,\Gamma} &= \sum_{\gamma}p_\gamma \ket{\Phi_\gamma}\bra{\Phi_\gamma} \otimes \rho_{LRM,\mathrm{rest}}^{(\gamma)} \otimes \ket{\gamma}\bra{\gamma} \label{eq: full decomposition}
\end{align}
where $\rho_{LRM,\mathrm{rest}}^{(\gamma)}=\sum_{m|\gamma} p_{m|\gamma}\ \rho_{L,\mathrm{rest}}^{(m,\gamma)} \otimes \rho_{R,\mathrm{rest}}^{(m,\gamma)} \otimes\ket{m}\bra{m}$, and $\Gamma$ is collection of all possible values of $\gamma$. Since the above density matrix is decoupled, we either measure or trace out the non-logical degree of freedoms to get \begin{align}
    \rho_{L,R,\Gamma} = \sum_{\gamma}p_\gamma \ket{\Phi_\gamma}\bra{\Phi_\gamma} \otimes \ket{\gamma}\bra{\gamma}
\end{align} and the coherent information accessible using $\Gamma$ is given by
\begin{align}
    I_c^{(0)}(L:R\Gamma) &= S(R\Gamma) - S(LR\Gamma)\nonumber \\
    &= \sum_\gamma p_\gamma S(\mathrm{Tr}_L \ket{\Phi_\gamma}\bra{\Phi_\gamma}),\\
    &= I_c(L:RM)\nonumber
\end{align}

{Coherent information is maximum for when each post-measurement state in logical space is maximally entangled. Also, it was argued in~\cite{Marvian_2017} that for SPTs protected by Ableian symmetry, $\gamma$ is given by the symmetry charge in the bulk, that is, for Abelian symmetry, $\gamma = \prod_{i}m_i$. This is intuitive since to have a stable notion of quantum phase for quantum communication $\gamma$ has to be robust against weak symmetric perturbation in the bulk and should be a function of the bulk symmetry charge.

From now on, we take $\gamma$ to mean the symmetry charge in the bulk. This also implies that any symmetric perturbation in the bulk should not affect the information transmission as the value of $\gamma$ does not change under the perturbation. We will come back to this point later when we talk about moving away from the fixed point.}

\subsection{Examples of Protocols Using Pure SPT's}\label{sec: examples pure case}
We now provide concrete examples to illustrate these ideas:

\emph{Cluster State in $d=1$:} The Hamiltonian of the 1d cluster state~\cite{raussendorf2003measurement}, which describes a zero-correlation-length limit of a one-dimensional SPT phase protected by an on-site $Z_{2}\times Z_{2}$ symmetry is given by
\begin{equation}
    H = -\sum_{i} Z_{i-1}X_{i}Z_{i+1}.
\end{equation}
It is known as a $Z_2 \times Z_2$ SPT, where the two symmetry charges are $G_{0} = \prod_{i}X_{2i+1}$ and $G_{e} = \prod_{i} X_{2i}$. The string-order parameter that distinguishes the non-trivial $Z_2 \times Z_2$ SPT state to a $Z_2 \times Z_2$ symmetric trivial state is given by
\begin{equation}
    \mathcal{S}_e = Z_{2i} \prod_{k = i}^{j - 1} X_{2k + 1} Z_{2j} = 1,
\end{equation}
with similar definition of $S_o$ whose two end points live on the odd sublattice. When the system ends on site $L$ and site $R$ (both living on the odd sublattice), we keep terms in the Hamiltonian that are fully supported in the bulk. Then, we can interpret $L$ and $R$ as our logical subspaces, and the logical operators are given by $X_{L(R)}$ and $Z_{L(R)}$. \newline

After measuring qubits in bulk with measurement outcome $m$, the post-measurement state is given by,\begin{align}
     \ket{\psi_{LR,m}} = Z_L^{\gamma_{e}}X^{\gamma_{o}}_L \frac{\ket{11}+\ket{00}}{\sqrt{2}}, \label{eq: 1d cluster state decoding}
 \end{align} 
where $\gamma_{e} = \prod_{i} m_{2i}$ and $\gamma_{o} = \prod_{i=i}m_{2i + 1}$ are nothing but the measurement outcome of the bulk $Z_2 \times Z_2$ symmetry charges. Together $(\gamma_e,\gamma_o)$ identify the type of entangled pair in the logical space and hence is enough to transmit the information. We have $I_c(L:R\Gamma)=1$ where $\Gamma=(\gamma_e,\gamma_o)=(1,1),(1,0),(0,1),(0,0)$ are the possible values of $\gamma$. 

\begin{figure}
    \centering
      \includegraphics[width=0.5\textwidth]{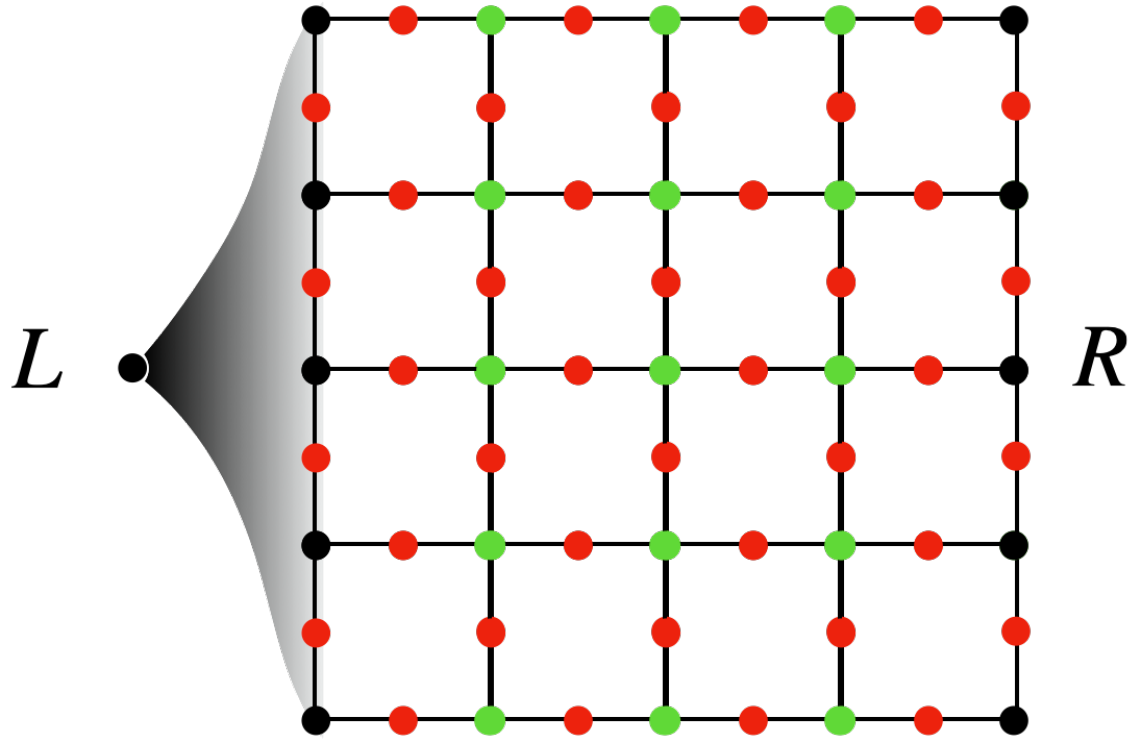}
    \caption{The qubit $L$ is entangled with a subspace of the left boundary of the 2d cluster state, corresponding to the codespace of a one-dimensional repetition code. The repetition code encoding is used in hindsight, with the knowledge that the cluster state with bulk measurement transmits the repetition code space. However, we note that the probe of quantum information transfer between $L$ and $R$ does not require prior detailed knowledge of this codespace. For example, extensively many ancillas can be entangled with the left boundary, to detect the coherent transmission of a qubit of quantum information through the bulk.   }
    \label{fig:2d mbqc}
\end{figure}

\emph{Cluster State in $d=2$:} The 2d cluster state~\cite{raussendorf2005long} describes the zero-correlation-length limit of a two-dimensional SPT phase protected by an on-site, global $Z_{2}$ symmetry, along with a one-form $Z_{2}$ symmetry.  This state is defined on two dimensional Lieb lattice (Fig.~\ref{fig:2d mbqc}) with qubits defining on both vertices and edges. The Hamiltonian of the 2d cluster state is given by
\begin{equation}\label{eq: Hamiltonian of the 2d cluster state}
    \begin{split}
        H = -\sum_{v} X_{v} \prod_{e \ni v} Z_e - \sum_{e} X_e \prod_{v \in e} Z_v.
    \end{split}
\end{equation}where $v$ labels qubits on the vertices and $e$ labels qubits on the edges. The 2d cluster state is known as a $Z_2^{(0)} \times Z_2^{(1)}$ SPT~\cite{yoshida2016topological}, where the upper index labels $0$-form and $1$-form symmetry. The symmetry action of $Z_2^{(0)}$ is given by $G_V = \prod_{v\in V} X_v$ where $V$ is the set of all vertices, and the symmetry action of $Z_2^{(1)}$ is given by $G_E = \prod_{e\in C} X_e$ where $C$ are the edges on any closed loop on the lattice. The 2d cluster state admits a string order parameter and a membrane order parameter, which together distinguish it from a trivial symmetric state. Those two order parameters are given by
\begin{equation}
    \begin{split}
        &\mathcal{M} = \prod_{e \in \partial A} Z_e \prod_{v \in A} X_v = 1\\
        &S = Z_{v_i} \prod_{e \in l} X_e Z_{v_j} = 1
    \end{split}
\end{equation}
where $\partial A$ is a closed loop on the dual lattice enclosing $A$ and $l$ are edges on a string living on the direct lattice whose endpoints are $v_i$ and $v_j$. We put the 2d cluster state on an open cylinder with periodic boundary conditions in the vertical direction and open boundary conditions in the horizontal direction; the open boundary is terminated as shown in Fig.~\ref{fig:2d mbqc}. We keep terms in the Hamiltonian that are fully supported in the bulk. We identify the two 1-dimensional boundaries as logical space $L$ and $R$. Note that the edge qubits at the boundaries are considered part of the bulk. The corresponding logical operators are $\Bar{X}_{L(R)} = \prod_{v \in V_{L(R)}} X_v$ and $\Bar{Z}_{L(R)} = Z_{v_{L(R)}}$, where $V_{L(R)}$ denotes the set of vertex qubits on the $L(R)$ and $v_{L(R)}$ denotes any one of the vertex qubits on $L(R)$. The logical operators can be verified by finding restrictions of the symmetry operators at the boundaries. Another way to check these is as follows:
 After measuring the bulk qubits (red and green qubits in Fig~\ref{fig:2d mbqc}) the two boundaries are entangled such that $\langle \prod_{i\in V_L}X_i \prod_{j\in V_R}X_j \rangle = \prod_{i\in V} m_i$ and for any $Z_i$ and $Z_j$ on the vertices of $L,R$ boundary respectively, $\langle Z_i Z_j\rangle = \prod_{\l \in \ell_{i,j}} m_l$, where $V$ denote the vertices and $\ell_{i,j}$ is a string passing through edges connecting the $i,j$ sites at the opposite boundaries. These can be checked by multiplying appropriate stabilizers of the cluster state. The logical operators are therefore $\prod_{i\in V_L} X_i$ and $Z_i$ for any $i\in V_L$ since they develop long-range correlations and are non-commuting. The logical space is the same as the repetition code which allows us to rotate the boundaries and trace out the non-logical qubits. To decode the logical information, we only need to know $\gamma_V = \prod_{i\in V}=m_i$ and $\gamma_\ell = \prod_{l\in \ell}m_l$ where $\ell$ is a string passing through links connecting the two logical qubits.

\section{Information Transmission using Mixed SPT Orders}\label{sec: Information transmission using Mixed SPT}

In this section, we use the ideas presented above to define mixed SPT states. We start with a pure SPT and put local decoherence. The decoherence can be thought of as environment qubits coming close and interacting with the system's qubits. We take the environment to be initialized in a product state and the system-environment interaction to be local. We denote the environment by $E$ and the system by $S$. As above, $L,R$ denotes the left and the right boundary of the system. 

For mixed states, we study the information about $L$ at two locations, 1) in $E,R$ and measurement outcomes $m$ combined, $I_c(L:ERM)$, or 2) in $R$ and $m$ (without needing access to the $E$ qubits), $I_c(L:RM)$. As we will see in the next section, this classification of mixed-state SPT based on the above information theoretic probes is also related to other mixed-state SPT probes such as strange correlators. Surprisingly, we find a class of channels dubbed as \condition channels that destroy the symmetry in the system but have non-trivial behavior for the above quantities. These channels are outside the scope of any studies on mixed SPTs so far.
In Sec.~\ref{sec: I_c(l:ERM)} we explore the conditions for the measurements not destroying the information. In other words, when can the information be recovered with the access to the environment? We prove Theorem~\ref{thm 1} which gives the sufficient condition for such. We also introduce various types of decoherence channels. In Sec.~\ref{sec: I_c(L:RM)} we explore the behavior of the information without access to the environment, $I_c(L:RM)$. We derive the expression in eq.~\eqref{eq: I_c(L:RM) using H(p)} for the coherent information and relate it to the performance of a decoder in learning the symmetry charge of the system. We also present calculations for $I_c(L:RM)$ for various channels.

Before we move our discussion further, we want to review the standard symmetry conditions for a density matrix $\rho$. Given a symmetry group $G$, we say $G$ is a strong symmetry if any group element $g \in G$ acts as
\begin{equation} \label{eq: strong symmetry}
    U_g \rho = e^{i \theta(g)} \rho,
\end{equation}
where $e^{i\theta(g)}$ is a phase factor. We say $G$ is a weak symmetry if any group element $g \in G$ acts on a symmetric state as
\begin{equation} \label{eq: weak symmetry}
    U_g \rho U_g^{\dagger} = \rho.
\end{equation}
We call a quantum channel $\mathcal{E}$ a strongly/weakly symmetric channel if $G$ is a strong/weak symmetry for $\mathcal{E}[\rho]$. We will also consider \condition channels later which are neither strong nor weak symmetric. In the rest of the paper, we always assume $G$ to be the symmetry group or a subgroup of the symmetry group of a pure SPT phase. 

\subsection{Coherent Quantum Information with Access to the Environment, $I_c(L:ERM)$}\label{sec: I_c(l:ERM)}
Let the system be an SPT state protected by on-site symmetry $G=\otimes_{i} G_i$. We can purify the decoherence channel $\mathcal{E}[\cdot]$ that acts on the system $\ket{\psi}\bra{\psi}$ by introducing an environment state $\ket{E}$ and system-environment interaction $U^{SE}$.  The channel can then be represented as
\begin{equation}
    \mathcal{E}[\ket{\psi}\bra{\psi}] = \Tr_{E}[U^{SE}\ket{\psi}\otimes \ket{E}\bra{E} \otimes \bra{\psi}(U^{SE})^{\dagger}].
\end{equation} In the composite pure state of system and environment $\ket{\psi^{SE}} (\equiv U^{SE}\ket{\psi}\otimes \ket{E}$), the original symmetry charge $G$ of $\ket{\psi}$ gets transformed into $G^{SE} \equiv U^{SE}G (U^{SE})^{\dagger}$. {Note that we are abusing the notation by also using $G$ for the set of unitary operators representing the action of symmetry group $G$. That is, $G$ also represents elements from the set $\{U_g | g\in G\}$ where $U_g$ is the representation of $g$.}

The above symmetry conditions for a quantum channel $\mathcal{E}$ can be restated in terms of its purification. If $\mathcal{E}$ is strongly symmetric, then the new symmetry charge can only take the form of $G^{SE} = G \otimes \mathds{1}^E$. For weakly symmetric $\mathcal{E}$ there exists a purification of environment state $\ket{E}$ and interaction $U^{SE}$ such that the new symmetry charge takes the form of $G^{SE} = G \otimes G^{E}$ and $\ket{E}$ is symmetric under $G^{E}$\cite{ma2023average, de2022symmetry}.

In the rest of the paper, we focus on decoherence composed of single-qubit decoherence channel. This condition can be relaxed but we assume them to simplify the discussion. As a result, we can always choose an environment state $\ket{E}$ which is a product state like $\ket{E} = \otimes \ket{e_i}$. Accordingly, the system-environment interaction $U^{SE}$ can be chosen as the tensor product of local uniary gates, i.e. $U^{SE} = \otimes U^{SE}_i$. Each local symmetry charge $G_i$ then gets transformed into $G^{SE}_i = U^{SE}_i G_i (U^{SE}_i)^{\dagger}$.  We always want to measure the system qubits in eigenbasis $G_i$. Without decoherence, local symmetry charge measurements of the system can transmit the quantum information as discussed in Sec.~\ref{sec: MBQC pure}. After interaction with the environment, we can still transmit coherent information by measuring $G^{SE}_i$ in the bulk. 
However, to call the mixed state an SPT state we still want to measure $G_i$ (or any other symmetric operator) on the system otherwise away from the fixed point the measurement operator on the system might change. Such a protocol to characterize the SPT phase will then be ill-defined since we could also perform measurements on special non-SPT states and entangle the boundaries~\cite{Popp_2005}. This restriction on the measurement protocol leads to the following condition for the measurements to not destroy the information, $G^{SE}_i P_{m_i} = O^{(m)}_i$, where $P_{m_i}$ is a projection operator to an eigenstate of $G_i$ and $O^{(m)}_i$ is some operator supported on the environment alone and depends on the measurement outcome $m$. 
Another way to write the above condition is $G^{SE}_i=\sum_m O^{(m)}_i P_{m}$. In other words, $G^{SE}_i$ have product eigenstates $|o^{(m)}_i\rangle\otimes \ket{m_i}$ where $|o^{(m)}_i\rangle$ is an eigenstate of $O_{i}^{(m)}$, so that measuring $G_i$ is not incompatible with measurement of the true symmetry charge $G_i^{SE}$. We call this condition the \textit{\condition} condition. It implies that the value of $G^{SE}_i$ can be determined by first measuring $G_i$ on the system followed by measuring $O^{(m)}_i$ on the environment (the measurement of the environment depends on the measurement outcome of the system).

A couple of things to note before we move on. The \condition channels above guarantee transmission of quantum information when the system is measured in the local symmetry charge basis (as we will show in the theorem below). Here we assume that measurements on the system are performed before the measurements on the environment. In general, this sequence is important as the measurement basis for the environment depends on the measurement outcomes of the system. Also, for SPTs with higher form symmetries, $I_c(L:ERM)$ may be positive even for weak non-symmetry-decoupling channels and show a transition to zero value as the decoherence strength is increased. We note that requiring $I_c(L:ERM)$ to be positive for any strength of decoherence requires that this decoherence channel obeys the symmetry-decoupling condition.

Let us be more precise now.
Let $m$ be local measurements performed on the system in the symmetry basis. Let us denote the set of all possible measurement outcomes by $M$. Let $I_c^{(0)}(L:RM)$ be the coherent information without decoherence. We then have the following theorem.

\begin{theorem} \label{thm 1}
Let new symmetry operator post-decoherence supported on the system and the environment be of the form $G^{SE} = \sum_{m\in M} O_E^{(m)} P_{m}$, where $P_{m}$ is the projection to measurement outcome $m$ in the bulk and $O_E^{(m)}$ are operators on the environment. The coherent information post-decoherence is then the same as that without decoherence, $I_c(L:ERM)=I_c^{(0)}(L:RM)$, where $I_c^{(0)}(L:RM)$ is the coherent information without decoherence.
\end{theorem}

\begin{proof}
    Let $\rho_0 = \ket{\psi}\bra{\psi}$ be the pure SPT state. Let $\rho^{(m)} \propto P_m \rho_0 P_m = \rho_{LR}^{(m)} \otimes \ket{m}\bra{m}$ be the post-measurement state of the boundaries. $P_m$ is projection on the system to outcome $m$. $I_c^{(0)}(L:RM)$ is given by the average of von-Neumann entanglement entropy between $LR$, $I_c^{(0)}(L:RM)= \sum_m p_m S_R(\rho_{LR}^{(m)})$ (see Appendix~\ref{app: i_c}).

    Post-decoherence, in addition to the system, measurements on the environment are also considered. In particular, the environment is measured in basis $O^{(m)}_E$ which depends on the measurement outcome $m$ on the system. 
    Let us denote the environment measurements by $e$ and the set of outcomes by $M_E$. The post-measurement state of the boundaries is given by $\rho_{LR}^{(m,e)} = P^E_e P_m U\rho_0 \rho_0^EU^\dagger P_m P^E_e$, where $U$ is the interaction between $E,S$ and $\rho_0^E$ is the initial state of the environment. 
    $P_e^E$ is the projection of environment to eigenvalue of $O^{(m)}_E$. Note that the projectors without any superscript denotes projection on the system and those with superscript $E$ and $SE$ are projectors on the environment and the union of system and environment respectively.
    The \condition condition for the symmetry $G^{SE}$ implies $P^E_e P_m = P_m P^{SE}_e$ where $P^{SE}_e$ is projection to eigenstates of $G^{SE}$ with eigenvalue $e$. Using $G^{SE} = UGU^\dagger$ we also have $P_e^{SE}U = UP_e$ and the post-measurement state is $P_m U P_e\rho_0 \rho^E_0P_e U^\dagger P_m = P_m U \left(\rho_{LR}^{(e)} \otimes \ket{e}\bra{e}\right) \rho_0^E U^\dagger P_m = \rho_{LR}^{(e)}\otimes \rho_{\mathrm{junk}}$, where $\rho_{\mathrm{junk}} = P_m U \left(\ket{e}\bra{e}\rho_0^E\right) U^\dagger P_m$ (recall that we assume that $U$ doesn't act on the boundaries). Since the reduced density matrix on $LR$ post-measurement is the same as that without the decoherence, the coherent information is unchanged, $I_c(L:ERM)=I_c^{(0)}(L:RM)$. This completes the proof.
    
\end{proof}

After measurements on $E,S$ the state on the boundaries and the measurement register is \begin{align}
    \rho_{L,R,M,M_E} = \sum_{m\in M} p_m \sum_{e\in M_E} p_{e|m} \rho_{L,R}^{(m,e)} \otimes \ket{m,e}\bra{m,e},
\end{align}
where $p_m,p_{e|m}$ are probability to observe outcomes $m$ on system and $e$ on environment given that system had outcomes $m$ respectively.
The coherent information in $E,R,M$ is given by,\begin{align}
    I_c(L:ERM) &= \sum_m p_m \sum_e p_{e|m} \left(S(\rho_R^{(m,e)}) - S(\rho_{LR}^{(m,e)})\right) \nonumber \\
    &= \sum_m p_m \sum_e p_{e|m} S(\rho_R^{(m,e)}),
\end{align}
and is be equal to the information before the decoherence as the reduced density matrix on the boundaries is unchanged.

Following the discussion around eq.~\eqref{eq: full decomposition}, we don't need full knowledge of the measurement outcomes to transmit the logical space and only the equivalency class $\gamma$ of the post-measurement logical state is required. We can throw away the bulk measurements on $E$ except for the symmetric charge $\gamma_e=\prod_{i} e_i$ and trace/measure out the non-logical space on the boundaries.

We also note that though the above formalism though seems to depend explicitly on how the environment is chosen, coherent information does not depend on such choice. Using Stinespring theorem a minimal p We can apply any unitary rotation on the environment but this only changes the new symmetry charge $G^{SE}$ or more precisely, the operators $O_i^{(m)}$. The crucial point is that there always exists a measurement basis in the environment, which might be highly non-local, such that performing measurements on such a basis can transmit the information across the bulk. Another important point is that the measurement basis on the system is fixed, and it is the measurements in this symmetric basis that protect the information transmission away from the fixed point. This makes the above probe well-defined for the mixed SPT phase. We will come back to this point in detail when we perturb away from the fixed point in Sec.~\ref{sec: away from fixed point}.

We now discuss various examples of decoherence channels for the 1d cluster state to demonstrate the above ideas.
The 1d cluster state has $Z_2 \times Z_2$ symmetry generated by $\prod_{i \in A}X_i$ and $\prod_{i\in B} X_i$ where $A,B$ are two sublattices bi-partitioning the lattice. We can apply decoherence to both sublattices independently with different strengths, $p_a,p_b$. We now discuss a few types of decoherences.

\paragraph{\textbf{Z dephasing}}
A local Z-dephasing channel acts on the density matrix as\begin{align*}
    \mathcal{N}^Z_i[\rho] = (1-p)\rho + pZ_i\rho Z_i,
\end{align*}
where $p$ is the strength of the dephasing. We consider the channel has different strength $p_a$ and $p_b$ when acting on sublattices $A$ and $B$. The channel can be purified by introducing an environment initialized in a product state of $\ket{0}$s and the system-environment interaction is given by $U=\mathrm{CZ.Rx}(\theta)$ where CZ is the controlled-Z gate and Rx$(\theta)$ is rotation around the x-axis by angle $\theta$. The strength of the dephasing is given by $p=\sin^2 \theta$ and $\theta\in [0,\pi/4]$. The dephasing channel does not act on the system boundary qubits. The new symmetry charges after the interaction are $\prod_{i\in A(B)}Z^E_i X_i$. Notice that the new symmetry charge is independent of the dephasing strengths $p_a$ and $p_b$.

Since the new symmetry is of the form required by Theorem~\ref{thm 1}, the coherent information $I_c(L:ERM)=1$ since the pure cluster state can transmit information perfectly, see eq.~\eqref{eq: 1d cluster state decoding}.
Performing $X_i$ measurements in the bulk of the system followed by $Z^E_i$ measurements in the environment allows the symmetry charge to be learned and is enough to transmit the information following discussion around eq.~\eqref{eq: 1d cluster state decoding}, the information transmitted after post-decoherence with the help of system and environment is equal to that pre-decoherence.

The same analysis and conclusion also hold for Y-dephasing.
\newline
\paragraph{\textbf{SWAP channel}}
The SWAP channel is defined as a SWAP gate acting between system qubits and the environment qubits initialized in the product state of eigenstates of X operator $\ket{+}$ state. The Krauss operator representation is given by \begin{align}
    \mathcal{N}^{\rm Sw}_i[\rho]  = P_{i,+} \rho P_{i,+} + Z_iP_{i,-} \rho P_{i,-} Z_i,
\end{align}
where $P_{i,\pm}$ are projection to $\pm$ eigenstates of $X_i$. Clearly since the environment has access to the original qubits from the system\begin{align}
    I_c(L:ERM)=1. \label{eq: Ic(ERM) swap channel}
\end{align}
This can also be argued using Theorem~\ref{thm 1} as the new symmetry charge is $\otimes_i X_i^E$ and satisfy the \condition condition.
\newline

\paragraph{\textbf{Controlled Hadamard}}
The controlled Hadamard gate is defined by initializing the environment in $\ket{0}$ state and applying interaction $U=\mathrm{CH.Rx}(\theta)$ where CH is controlled Hadamard gate with control from the environment. The new symmetry charge on site $i$ is given by $UX_i U^\dagger \propto {Z^E_i Z_i H_i Z_i + H_i}$. Let the decoherence act only on one sublattice. Measuring the system in X-basis would destroy $UX_i U^\dagger$ as they don't commute. Thus we expect \begin{align}
    I_c(L:ERM) = 0.\label{eq: Ic(ERM) for controlled H}
\end{align}
For decoherence acting on both sublattices the coherent information $-1$.
\newline

\paragraph{\textbf{\Condition channel}} \label{non weak-symmetric paragraph}
Let us consider a channel where the
new local symmetry charge under the generic \condition channel is given by $G^{SE}_i = \sin\theta X_i X^{E}_i + \cos\theta I Z^E_i$, where $X^E,Z^E$ acts on the environment. This channel, for example, can be purified by introducing the interaction \begin{align}
U^{SE}_i=CNOT\cdot e^{i\theta Y^E_i}\cdot SWAP \label{eq: SDC U}
\end{align}
where the CNOT is controlled by the environment qubit, and an environment state $\ket{e_i} = \cos\phi \ket{0^E} + \sin\phi \ket{1^E}$. 

In general, the channel defined via the above interaction is not weakly symmetric except at special points, $\cos\theta=0$ or $\tan \phi=\pm 1$; the latter is true when the environment starts in the eigenstate of $X_i^E$. The proof can be found in Appendix~\ref{app: SDC under symmetry}.

Using Theorem~\ref{thm 1} the coherent information $I_c(L:ERM)=1$. We will later see that the coherent information $I_c(L:ERM)$ is closely connected to the mixed-state strange correlator and density matrix under \condition channel has a non-zero value for the strange correlator despite being a non-weak-symmetric channel. This is an interesting result as so far the studies of mixed-state SPT were restricted to weak-symmetric channels and the above channel is an example of a non-weak-symmetric channel with SPT order in the mixed-state.

\subsection{Coherent Information without access to the environment, $I_c(L:RM)$}\label{sec: I_c(L:RM)}

Here we ask about the amount of coherent information transmitted to $R$ without having access to the environment. The new symmetry charge $G^{SE}$ cannot be deterministically known if we don't have access to the $E$ qubits. The best one can do in this case is to ``guess'' or \textit{decode} the values of $\gamma_e$ using the measured outcomes. The coherent information stored in the right boundary $R$ and the measurement outcomes $M$, $I_c(L:RM)$, is given by,\begin{align}
    I_c(L:RM) = \sum_m p_m S(\rho_R^{(m)}) - S(\rho_{LR}^{(m)}),
\end{align}
where \begin{align}
    \rho_{LRM} &= \sum_{m} p_m \rho_{LR}^{(m)} \ket{m}\bra{m}, \\
    \rho_{LR}^{(m)} &= \sum_{\gamma\in \Gamma_e}p_{\gamma|m} \sum_{e}p_{e|\gamma,m} \rho_{LR}^{(m,e,\gamma)} \nonumber \\
    &\equiv \sum_{\gamma\in \Gamma_e} p_{\gamma|m}\ket{\Phi_{\gamma,m}}\bra{\Phi_{\gamma,m}} \sum_{e} p_{e|\gamma,m} \rho_{LR,\mathrm{rest}}^{(m,e,\gamma)},
\end{align}
where we have split the boundary space into logical and non-logical degrees of freedom; the logical part only depends on $\gamma$. See discussion around eq.~\eqref{eq: logical space decomposition},~\eqref{eq: full decomposition}.
We measure or trace out the non-logical degree of freedom and have \begin{align}
    I_c(L:RM) &= \sum_m p_m S\left(\sum_{\gamma}p_{\gamma|m} \mathrm{Tr}_L \ket{\Phi_{\gamma,m}}\bra{\Phi_{\gamma,m}}\right) - \nonumber\\
    &\ \ \ \ - \sum_m p_m S\left(\sum_{\gamma}p_{\gamma|m} \ket{\Phi_{\gamma,m}}\bra{\Phi_{\gamma,m}}\right) \nonumber \\
    &\geq \sum_{m,\gamma\in \Gamma_e} p_{m,\gamma}S(\rho_R^{(m,\gamma)}) - \sum_{m}p_m H(p_{\gamma|m}),\nonumber\\
    &= I_c(L:E RM) - \sum_m p_m H(p_{\gamma|m}),
\end{align}
where $H(p_{\gamma|m})$ is the Shannon entropy of the distribution $p_{\gamma|m}$. In the above expressions we used the inequality \[\sum_ip_i S(\rho_i)\leq S(\sum_i p_i \rho_i)\leq \sum_ip_i S(\rho_i) + H(p_i)\] to go from line 2 to 3. If $H(p_{\gamma|m})=0$ then $I_c(L:RM) = I_c(L:E RM)$ (since $I_c(L:RM)\leq I_c(L:ERM)$ by data processing inequality~\cite{Schumacher_1996}). In other words, the information can be transmitted without access to the environment qubits if the symmetry charge leaked into the environment can be perfectly learned from the measurement outcomes $M$ on the system. 

Also interestingly, the inequality in the above expression is saturated when $\ket{\Phi_{\gamma,m}}$ are orthogonal for different $\gamma$ and $I_c(L:ERM)$ is maximum. That is, if $\braket{\Phi_{\gamma_1,m}}{\Phi_{\gamma_2,m}} = \delta_{\gamma_1,\gamma_2}$ and $\mathrm{Tr}_L \ket{\Phi_{\gamma,m}}\bra{\Phi_{\gamma,m}}\propto \mathbb{I}$ then \begin{align}
    I_c(L:RM)=I_c(L:ERM)-H(p_{\gamma|m}).  \label{eq: I_c(L:RM) using H(p)}
\end{align}
The above equation is one of key results of the paper. We expect the conditions for the above equation to be true for computationally complete resource states such as cluster states. The above equation thus provides an important information-theoretic identity for the computational use of mixed-state SPTs.

We now work out a few examples for cluster states in various dimensions.
\begin{figure}
    \centering
\includegraphics[width=\linewidth]{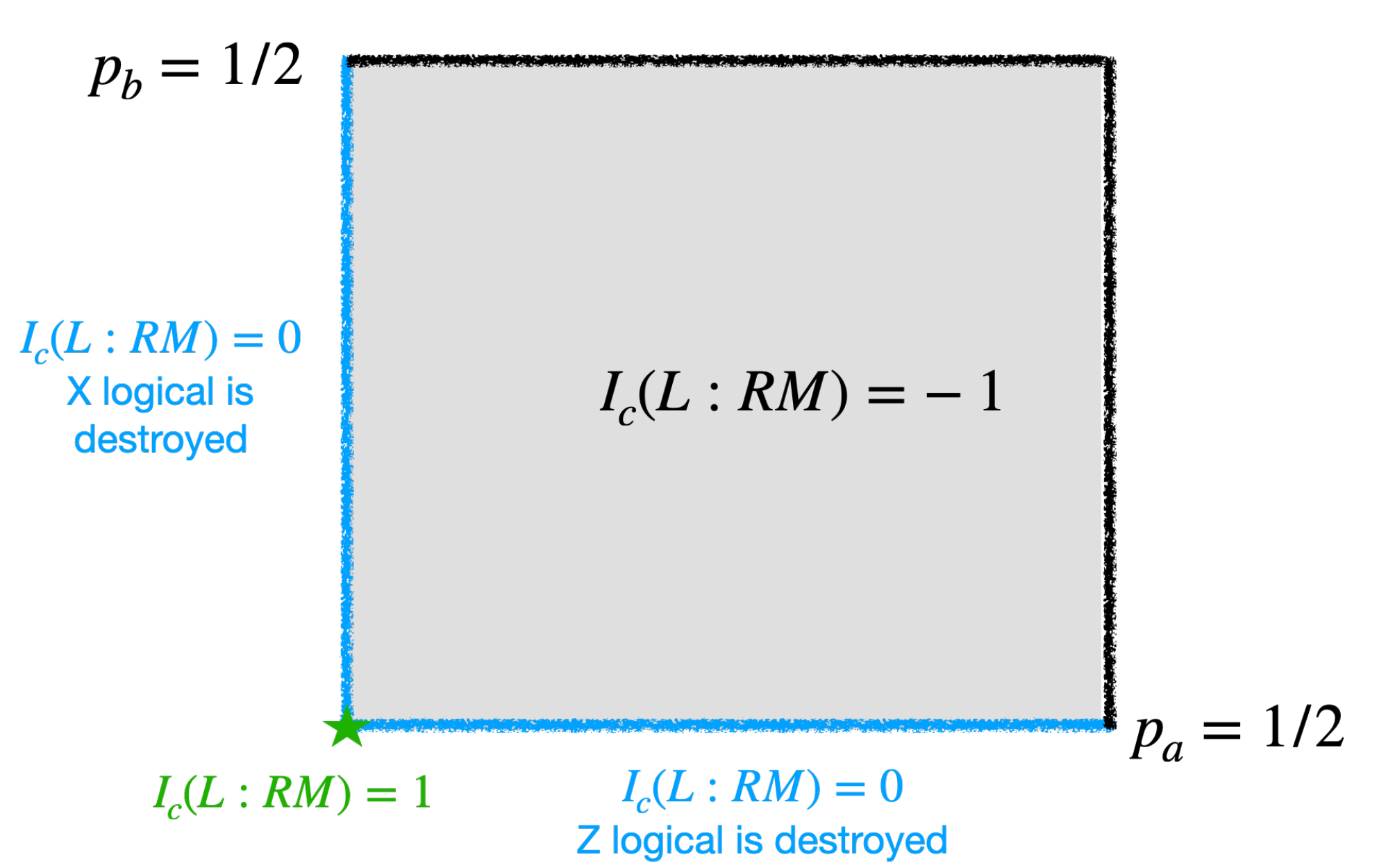}
    \caption{Behavior of $I_c(L:RM)$ of 1d cluster state under decoherence. The two axes are the strength of the decoherence on $a$ and $b$ sublattice. For blue regions only classical information is transmitted. When decoherence acts on both sublattices no information is transmitted (grey region).}
    \label{fig:1d cluster state phase diagram}
\end{figure}

\subsubsection{Cluster State in $d=1$}
\paragraph{\textbf{Z-dephasing}}
For 1d cluster state, $\ket{\Phi_\gamma}$ are maximally entangled orthogonal states and thus $I_c(L:RM)=I_c(L:ERM)-\sum_m p_m H(p_{\gamma|m}) = 1-\sum_m p_m H(p_{\gamma|m})$. For 1d cluster state $\gamma=\prod_{i}Z^E_i\otimes X_i$ where $\gamma$ consist of two values $\gamma_0,\gamma_1$ defined for odd and even sublattices. Let us assume that the decoherence only acts on even sublattice for now, that is, $p_b=0$. Thus we need to estimate $\prod_{i} Z_i^E$ on even sublattice using measurement outcomes of $X_i$ on the system. 
The optimal strategy is to guess $Z_E^i$ using maximum likelihood given the value of $\langle Z_E^i\rangle = 1-2p_a $, i.e the probability of having $Z_i^E=1$ is $p_0 = \left(\langle Z_E^i\rangle + 1\right)/2=1-p_a$. 
Using the above maximum likelihood probability, the probability of $\prod_i Z_i^E = 1$ is $p_{\gamma=1|m}=\tfrac{1 + (1-2p_a)^N}{2}$ and $p_{\gamma=-1|m}=\tfrac{1 - (1-2p_a)^N}{2}$, where $N$ is the number of qubits getting decohered. 
The entropy of this distribution is \[\sum_m p_m H(p_{\gamma|m})\approx 1-(1-2p_a)^{2N}/(2\ln 2)\] and \begin{align}I_c(L : RM) =(1-2p)^{2N}/(2\ln 2)=e^{-N/\xi}\label{eq: Ic(RM) Z dephaing one sublattice}\end{align} with $\xi=\tfrac{-1}{2\ln (1-2p)}$.

If both the sublattices are  decohered then \[\sum_m p_m H(p_{\gamma|m})\approx 2-\left((1-2p_a)^{2N} + (1-2p_b)^{2N} \right)/(2\ln 2)\] and \begin{align}
I_c(L : RM) =-1 + \left((1-2p_a)^{2N} + (1-2p_b)^{2N} \right)/(2\ln 2). \label{eq: Ic(RM) Z dephasing}
\end{align}
Fig.~\ref{fig:1d cluster state phase diagram} shows the phase diagram of the cluster state under decoherence.

Another way to arrive at the above result is to calculate $I_c(L:RM)$ using \begin{align}
    I_c(L:RM) = \sum_m p_m I^{(m)}(L:R) - I_c(L:ERM),
\end{align}
where $p_m$ is probability of observing outcomes $m$ and $I^{(m)}(L:R)$ is mutual information between $L,R$ for outcomes $m$. We will come back to this in Sec.~\ref{sec: connection to type I SC} where we will connect the behavior of $I_c(L:RM)$ with type-I strange correlators.

\paragraph{\textbf{SWAP channel}}
Under the SWAP channel, the system qubits are taken by the environment. If only one sublattice qubits are swapped then we can still determine $\gamma$ on the other sublattice $H(p_{\gamma|m})=1$ and thus $I_c(L:RM) = 0$. But if qubits from both sublattices are swapped then $H(p_{\gamma|m})=2$ and $I_c(L:RM)=0$.
\begin{figure}
    \centering
    \includegraphics[width=\linewidth]{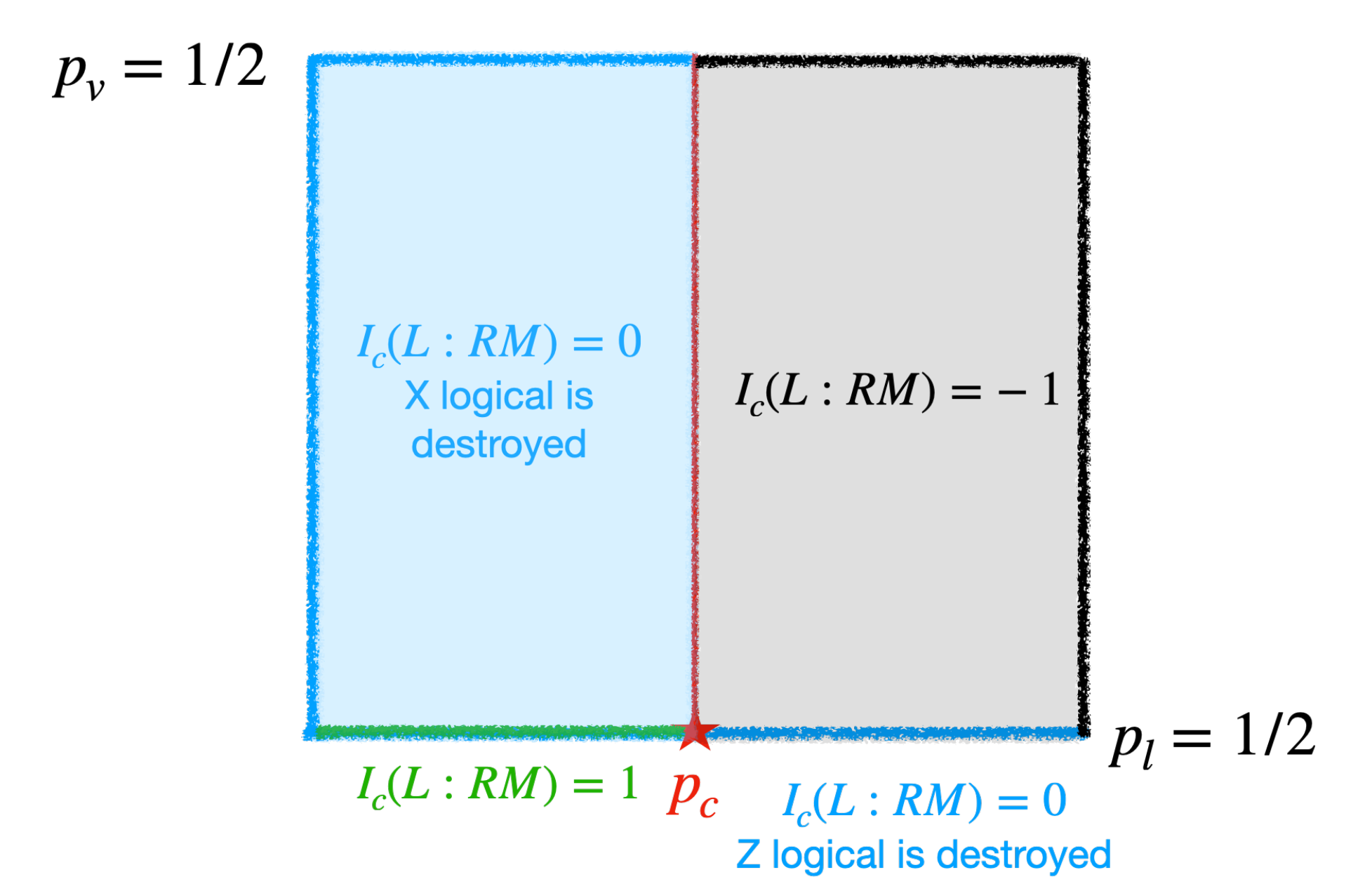}
    \caption{Behavior of $I_c(L:RM)$ of 2d cluster state under decoherence. $p_l$, $p_v$ are the strengths of the decoherence on link qubits and vertex qubits respectively. In the green region, the full logical information is transmitted, while only classical information is transmitted in the blue regions. The grey region denotes no information is transmitted and is completely lost to the environment. For $p_v=0$, $I_c(L:RM)$ sees a transition for \condition channels acting on links.}
    \label{fig: 2d cluster state phase diagram}
\end{figure}
\newline

\paragraph{\textbf{\Condition channel}} \label{sec: type I I_c non-weak-sym}
The calculation for generic \condition channel proceeds in a similar manner as for the Z-dephasing, except that instead of guessing the value of $Z_i^E$ we guess $O^{(m)}_i$. We can show that $\langle O^{(m)}_i \rangle = 0$ and using the maximum likelihood guess for $\gamma$, $H(p_{\gamma|m})=1$ (decoherence acts only on one sub-lattice) and $I_c(L:RM)=0$.

\subsubsection{Cluster State in $d=2$}
 \paragraph{\textbf{Z-dephasing}}
 As discussed in Sec.~\ref{sec: examples pure case}, for information transmission the values of $\gamma_V=\prod_{i\in V}m_i$ and $\gamma_l = \prod_{i\in \ell_H} m_i$ are needed where $m_i$ are measurement outcomes of $X_i$ in the bulk, and $\ell_H$ is a string passing through the links and ending on the two boundaries $L,R$. After interaction with the environment, the new on-site $Z_2^{(1)}$ symmetry charge is $UX_i U^\dagger = Z^E_i X_i$. In the absence of any decoherence on the vertices, $p_v=0$, $H(p_{\gamma_v|m})=0$ as all measurement outcomes on the vertices are trustful and we have $\gamma_V=\prod_{i\in V}m_i$. For decoherence on edges with strength $p_l$, the measurement outcomes on edges can be faulty and we want to guess/decode the true symmetry charge. By $Z_2^{(1)}$ symmetry, for any plaquette $\prod_{i \in \square} Z^E_i X_i = 1$ and therefore $\prod_{i\in \square} Z^E_i = \prod_{i\in \square} m_i$. We can think of the plaquette with $\prod_{i\in \square}m_i = -1$ as having an error syndrome, that is, one (or odd number) of the qubits in the plaquette has a faulty measurement outcome. We want to identify and remove these errors by applying strings operator $\prod_{i\in C} X^E_i$, where $C$ are open strings connecting the erroneous plaquettes and flipping the measurement outcomes $m_i$ along the string. The probability of an error occurring at site $i$ is equal to the probability of having $Z^E_i=-1$ and is equal to the decoherence strength $p_l$. Also not that, adding a contractible close loop of errors doesn't change the answer as the symmetry string will pass an even number of times through the error loop. The decoder thus need not identify the exact error string but a path equivalent to the original error string up to close loops of errors. Given the error syndromes, thus there are two equivalency classes for doing these matchings as shown below.\begin{align*}
    \centering
    \includegraphics[width=1\linewidth]{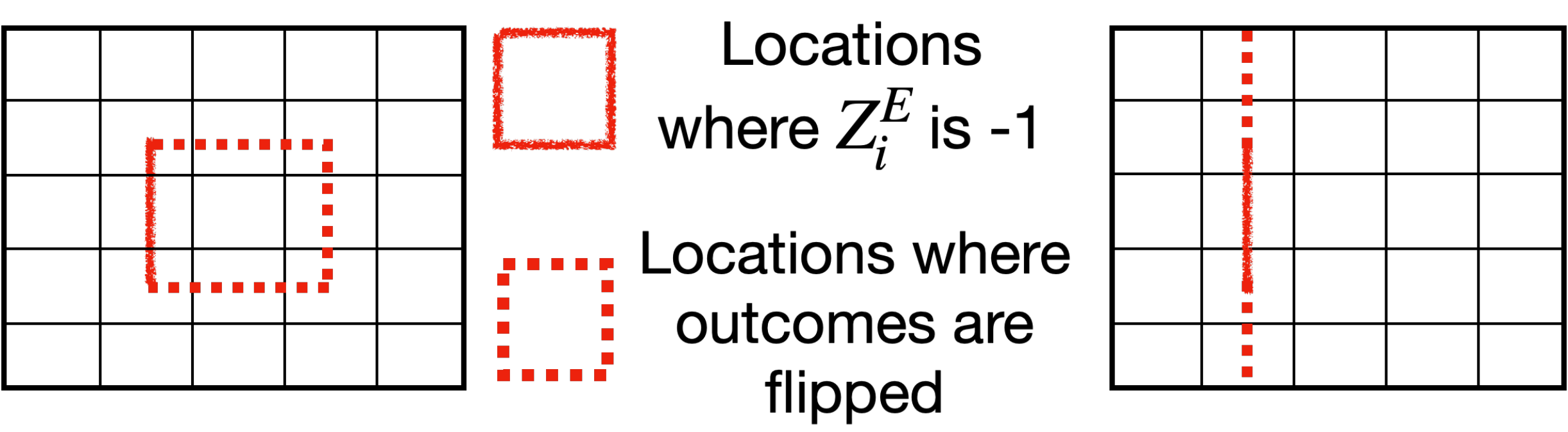}
    \label{fig:2d cluster state decoding}
\end{align*} 
A matching $C$ that does not create a non-contractible loop of $Z^E_i=-1$ (left figure) would then give the current value of $\gamma_l$ by $\prod_{i\in l} s_i m_i$, where $s_i=-1$ if $i\in C$ and is otherwise equal to $1$. However, if the probability of a non-contractible loop (right figure) is non-zero after matching of the errors, the value of $\gamma_l$ would be wrong and unreliable.
We also assume that there is no decoherence at the edge qubits on the boundaries otherwise there might be non-contractible loops of $\mathcal{O}(1)$ size starting and ending at the boundary. The above decoding problem is the same as the decoding problem for the Toric code under X-dephasing of strength $p_l$ and is known to have a finite error threshold~\cite{dennis2002topological} at $p_l\approx 0.109$. Moreover, the error threshold transition is known to lie in the random bond Ising model (RBIM) universality class along the Nishimori line~\cite{nishimori1981internal}. In conclusion, $H(p_{\gamma_l|m}) = 0$ for $p_l<p_c\approx 0.109$ and equal to $1$ for $p_l>p_c$ where $p_c$ is the critical point of 2d RBIM. And thus, $I_c(L:RM)$ behaves as \begin{align}
    I_c(L:RM) = \begin{cases}
                1 & p_l<p_c\\
                0 & p_l>p_c
                \end{cases}.
 \end{align}
The zero value of the coherent information is related to the fact that $\bar{X}$ logical operator is transmitted but the $Z$ logical is destroyed.

{On relaxing the assumption of no decoherence at the boundary edge qubits, there is a probability for a short non-contractible loop at the boundary. In this case, we expect $I_c(L:RM)<1$ but non-zero for $p_l<p_c$. We will study a similar scenario in Sec.~\ref{sec: connection to type I SC} where we will consider information transmission in SPT between any two points in the bulk.}

When $p_v$ is also non-zero then $H(p_{\gamma_v|m}) = 1 - \mathrm{exp}(O(N_V))$ where $N_V$ are the number of qubits on vertices getting decohered. This can be seen using the same line of reasoning as for the 1d cluster state above. In this case then \begin{align}
    I_c(L:RM) = \begin{cases}
                0 & p_l<p_c\\
                -1 & p_l>p_c
                \end{cases}.
 \end{align}
 We show the phase diagram for the 2d cluster state under decoherence in Fig.~\ref{fig: 2d cluster state phase diagram}.

\paragraph{\textbf{\Condition channel}}\label{sec: SDC calculation}

To get non-trivial behavior under \condition channels we need to consider channels smoothly connected to the identity channel. We therefore modify the channel such that with probability $1-q$ no decoherence occurs at the qubit and with probability $q$ the \condition channel defined in eq.~\eqref{eq: SDC U} acts. Then we show below that $I_c(L:RM)=1$ for low enough values of $q$.

Let us be more precise. For each decohered system qubit, we introduce two environment qubits $E_{1,2}$. The qubit $E_2$ is initialized in $\ket{e_2}=\sqrt{1-q}\ket{0}_2+\sqrt{q}\ket{1}_2$. We apply the following gate to the three qubits,\begin{align*}
    \mathbf{U}\ket{\psi}\ket{e_1}\ket{e_2} = \sqrt{1-q}\ket{\psi e_1}\ket{0}_2 + \sqrt{q}U \ket{\psi e_1}\ket{1}_2,
\end{align*}
where $U$ is the unitary introduced in eq.~\eqref{eq: SDC U} as \condition channel. The 2nd environment applies a control-$U$ gate on the system qubit and the 1st environment qubit. 

As discussed above, we want to know the symmetric charge string $\prod_{i\in \ell_h} X_i$ where $\ell_h$ is a string running from a qubit on $L$ to another qubit on $R$. Due to errors, the system outcomes would not give the right value for the above string. We define an outcome $m_i$ at bond $i$ to be erroneous if the actual symmetric charge at $i$ is $-m_i$. The probability of a bond being erroneous is given by $ \mathrm{Tr} \frac{1-m_iX'_i}{2} P_{m_i} \rho P_{m_i} \frac{1-m_iX'_i}{2}$, where $X'$ is the modified symmetry charge supported on the system and the environment qubit. The probability (up to normalization) for obtaining measurement outcomes $m$ and having erroneous qubits $\epsilon$ is given by,\begin{align}
    P(m,\epsilon) \sim &  \prod_{i\in \epsilon} q_i \prod_{i \notin \epsilon} (1-q + q_i)\ \  \times  \label{eq: disorder distribution main text}\\
    & \times \sum_{C} \prod_{i\in C} m_i {\epsilon_i},  \nonumber
\end{align}
where $C$ is a close loop on the lattice, $q_i =q \left(\frac{1+rm_i}{2}\right) $, $r=\langle X^E \rangle$ is the average value of the environment's qubits initial state, and $\epsilon_i = -1$ if $i$ is an erroneous qubit and $+1$ otherwise. Here we abuse the notation by denoting the set of erroneous qubits by $\epsilon$ and also $\epsilon_i = \pm 1$ are variables defined on the erroneous set; the use of the notation should be clear from the context.

From the probability distribution in eq.~\eqref{eq: disorder distribution main text} one can define a maximum likelihood decoder and write a stat mech model. Deferring the technical details to Appendix~\ref{app: decoder}, the decoder's behavior is determined by a disordered stat-mech model with partition function given by,\begin{align} 
    Z(m,\overline{\epsilon}) \sim \left(\sum_{\sigma} e^{ \sum_{ab} \beta J_{ab} \epsilon_{ab} \tau_a\tau_b + hm_{ab}}\right), \label{eq: stat mech model}
\end{align}
where $\tau_a$ are Ising spins living on the plaquettes of the original lattice. We have \begin{align*}
    e^{\beta J_{ab} + hm_{ab}} &\propto 1-q+q_{ab} \\
    e^{-2\beta J_{ab}} &= q_{ab}/(1-q+q_{ab}).
\end{align*}
Note that $q_{ab}$ and hence $J_{ab}$ depends on $m_{ab}$. The disorder distribution for $m_{ab}$ and $\epsilon_{ab}$ is given by eq.~\eqref{eq: disorder distribution main text}. {We are not able to exactly solve the model but we believe that model has an order-disorder transition as the value of $q$ is increased. This would correspond to having a transition in $I_c(L:RM)$ as the strength of the decoherence is increased. {One way to argue for a transition is to note that in the limit $\langle X_i^E\rangle = 0$ the stat mech model reduces to RBIM which is known to have order-disorder transition.} The universality class of the transition is however expected to be different than the RBIM transition in general. We also believe that not all weak symmetric channels will have RBIM transition and thus there is no connection between having weak-symmetry and having RBIM transition.}
    
\section{Information Transmission and the Mixed-State Strange Correlator}\label{sec: connection to SC}

In this section we make connections between the quantum communication property of mixed SPTs probed through coherent information $I_c(L:ERM)$ and $I_c(L:RM)$ and type-II and type-I strange correlators~\cite{lee2023symmetry,zhang2022strange} respectively. We show that positive $I_c(L:ERM)$ implies a non-zero type-II strange correlator and non-zero type-I strange correlator indicates transmission of non-trivial quantum information using the mixed state SPT, that is $I_c(L:RM)>0$. Strange correlators had been proposed as a many-body mixed state probe for mixed and average SPT based on order in double Hilbert space. Such probes usually lack practical and operational meaning. The results below provide an operational meaning to such probes. Moreover, as we have seen above, $I_c(L:ERM)>0$ for \condition channels, which need not be weakly symmetric. The results in this section then surprisingly imply that strange correlators are also non-zero for such channels.

\subsection{Type II strange correlator} \label{sec: connection to type II SC}
We prove that non-trivial the coherent information $I_c(L:ERM)$ implies non-trivial type-II strange correlator for mixed SPTs introduced in~\cite{lee2023symmetry}. The mixed state type-II strange correlator $\SC^{\rm{II}}$ is defined as \begin{align}
    \SC_{ij}^{\rm{II}} = \frac{\mathrm{Tr}\left( \rho O_i O_j \rho_0 O_i O_j^\dagger \right)}{\mathrm{Tr}\left( \rho \rho_0 \right)},
\end{align}
where $\rho$ is the density matrix of a decohered SPT state, whose symmetry group is $G$ before decoherence, and $\rho_0$ is a $G$-symmetric product state.
We purify the decoherence channel by defining an environment state $E$ and the system-environment interaction $U$. The composite state of the system and environment is another pure SPT state under the symmetry $G^{SE}=U G U^\dagger$. The type-I strange correlator for this pure SPT state $|\psi_{SE}\rangle$ is given by
\begin{equation}
    \SC_{ij}^{\RN{1}} = \frac{\mathrm{Tr}\left( \langle{\psi_{SE}} | C^{SE}_i{}^\dagger (C^{SE}_j) | {m,e^{(m)}}\rangle\right)}{\mathrm{Tr}\left(  \bra{\psi_{SE}} \ket{m,e^{(m)}}\right)},
\end{equation}
where $\ket{m,e^{(m)}} = \otimes_i |m_i,e^{(m)}_i\rangle$, and $\ket{m_i},|e^{(m)}_i\rangle$ are respectively eigenstates of $G_i$ and $O^{(m)}_i$ (recall by \condition condition $G^{SE}_i=\sum_{m_i} O^{(m)}_i P_{m_i}$, $\ket{m,e^{(m)}}$ is then a $G^{SE}$-symmetric product state), and $C_i^{SE}$ is charged under $G^{SE}$. 
We  write $C_i^{SE} = \sum_m C_i^{(m)}Z_i P_m$, where $Z_i\ket{m} = \ket{m'}$, $C^{(m)}_i$ is a unitary operator that rotates the environment qubit from eigenstate of $O^{(m)}$ to an eigenstate of $O^{(m')}$ 
such that $G^{SE}_i C_i^{(m)} Z_i |m_i,e^{(m)}_i\rangle = e^{i\theta} e^{(m)}_i |m'_i,e^{(m')}_i\rangle$ where $e^{i\theta}$ is a phase not equal to $1$. We can also show that \begin{align}
    G^{SE}_i{}^\dagger C_i^{SE} G^{SE}_i = e^{i\theta} C_i^{SE}. 
\end{align}

We introduce the following marginalized strange correlator which we will see below is related to the mixed-state type-\RN{2} strange correlator,
\begin{align}
    \SC_{ij}^{\RN{2}} = \frac{\mathrm{Tr}\left( \sum_{e^{(m)}} |\langle{\psi_{SE}} | C_i^\dagger Z_i C_j Z_j | {m,e^{(m)}}\rangle|^2\right)}{\mathrm{Tr}\left( \sum_{e^{(m)}}  |\bra{\psi_{SE}} \ket{m,e^{(m)}}|^2\right)}, \label{eq: SC pure modified}
\end{align} 
where we suppress the implicit superscript $(m)$ from operators $C_i$.
Since $\ket{\psi_{SE}}$ is a SPT state with symmetry $G^{SE}$, $|\bra{\psi_{SE}} C_i^\dagger Z_i C_j Z_j \ket{m,e}|^2 = |\bra{\psi_{SE}} \ket{m,e}|^2$ and the above expression for the modified \SC \ is non-zero.

Let us focus on the numerator. We can rewrite it as follows\begin{align*}
    &\sum_{e^{(m)}} |\bra{\psi_{SE}}C^\dagger_i Z_i C_j Z_j|{m,e^{(m)}}\rangle|^2 \\
    &= \sum_{e^{(m)}} |\bra{\psi_{SE}} Z_i Z_j |{m,e^{(m)}}\rangle|^2 \\
    & =  \rho Z_i Z_j \rho_0 Z_i Z_j,
\end{align*}
where $\rho = \mathrm{Tr}_E \ket{\psi_{SE}}\bra{\psi_{SE}}$ is the density matrix of the system. To go to the 2nd line above we have used the fact that $C_i$ is unitary rotation on $e_i$ and since we are summing over the complete basis of $e_i$, the overall rotation can be dropped. The denominator is similarly equal to $\mathrm{Tr}_S \rho \rho_0 $.  Thus we have, \begin{align}
    \frac{\mathrm{Tr}\left( \rho Z_i Z_j \rho_0 Z_i Z_j^\dagger \right)}{\mathrm{Tr}\left( \rho \rho_0 \right)} &= \frac{\mathrm{Tr}\left( \sum_e |\bra{\psi_{SE}} C_i C_j^\dagger \ket{m,e^{(m)}}|^2\right)}{\mathrm{Tr}\left( \sum_e  |\bra{\psi_{SE}} \ket{m,e^{(m)}}|^2\right)},
\end{align}
which is non-zero as argued in the above paragraph.

We have shown above that the symmetry-decoupling condition implies that the type-II strange correlator is non-zero. Moreover, $I_c(L:ERM)>0$ is true if and only if the decoherence is \condition. This concludes the claim that $I_c(L:ERM)>0$ implies a non-zero type-II strange correlator. We perform explicit calculations for strange correlators under \condition channel in Appendix~\ref{app: explicit SC calculation under SDC} and find non-trivial values for strange correlators.

\subsection{Type-I strange correlator} \label{sec: connection to type I SC}

In this section, we make the connection between coherent information $I_c(L:RM)$ and type-\RN{1} strange correlators. Since strange correlators are typically computed when the system is put on periodic boundary conditions, we would like to consider a different setup to better establish this connection. In particular, we would like to consider information transmission between two subsystems $L$ and $R$ for SPT with periodic boundaries (in the case of 1d and 2d cluster state, $L$ and $R$ each represent a qubit). We entangle one of the subspaces, say $L$, with a reference qubit in a Bell pair and ask how much coherent information is been transmitted to $R$ after measuring out the rest of the qubits; refer to Fig.~\ref{fig: MBQC setup} but with periodic boundary conditions. This is given by (see Appendix~\ref{app: i_c})
\begin{equation}\label{eq: coherence information in mixed trajectories}
    \begin{split}
        I_c(L:RM) =
    \sum_m p_m I^{(m)}(L:R) - S^{(m)}(\rho_L)
    \end{split}
\end{equation}
where $I^{\left(m\right)}(L:R)$ is the mutual information between $L$ and $R$ after measuring out the bulk qubits given measurement outcome $m$, $S^{(m)}(\rho_L)$ is the von-Neumann entropy of $L$ in the trajectory $m$, and $p_m$ is the probability of getting measurement outcome $m$ according to the Born rule.

\indent We first want to study the structure of the post-measurement reduced density matrix of $L$ and $R$ given by
\begin{equation}
    \rho^S_m \equiv \tr_{B}\left[ \rho^S \rho_m \right],
\end{equation}
where $B$ is the compliment of $L \cup R$. $\rho_m$ is a projection operator projecting to measurement outcomes $m$ in the region $B$ {(and can also be interpreted as the density matrix of a trivial product state labeled by $m$. This already hints the connection between strange correlator and measurement-based quantum communication.)} If we write $\rho^S_m$ in the symmetry charge basis, we can decompose it into a direct sum of different charge sectors:
\begin{equation}
    \begin{split}
        \rho^{S}_m = \bigoplus_{m_{LR}} \rho^{S}_{m, m_{LR}},
    \end{split}\label{eq: density matrix as tensor sum}
\end{equation}
where $m_{LR}$ labels the symmetry charge on $L$ and $R$. 

Let us focus on cluster state for simplicity. 
The type-I strange correlator is defined as \begin{align}
    SC^\mathrm{I}_{L,R} &= \frac{\mathrm{Tr}\rho^S Z_L Z_R \rho_{m}}{\mathrm{Tr}\rho^S  \rho_{m}}, \nonumber \\
    & = \frac{\mathrm{Tr}\rho^S_m Z_L Z_R \rho_{m_{LR}}}{\mathrm{Tr}\rho^S_m  \rho_{m_{LR}}}, \label{eq: type I SC reduced form}
\end{align}
where in the 2nd line we have traced out all quits except qubits in $L,R$.
It is obvious from the above expression that the off-diagonal elements of $\rho_{m,m_{LR}}^S$ control the behavior of the strange correlator (since $Z$ operator is the charge generating operator for $X$).
Thus, in the X-basis the reduced density matrices are  of the form \begin{equation}
     \rho^{S}_{m,m_{LR}} = p_{m_{LR}}
     \begin{pmatrix}
         \frac{1}{2} & \frac{1}{2}SC^\mathrm{I}_{L,R}  \\
         \frac{1}{2}SC^\mathrm{I}_{L,R} & \frac{1}{2} \end{pmatrix}, \label{eq: density matrix decomposition}
 \end{equation}
where $p_{m_{LR}}$ is the probability to observe charge $m_{LR}$. Note that we assumed that $SC_{L,R}$ does not change with $m_{LR}$ in eq.~\eqref{eq: type I SC reduced form} otherwise the diagonal elements won't be equal. From eq.~\eqref{eq: density matrix decomposition} and \eqref{eq: density matrix as tensor sum} we find that \begin{align}I^{(m)}(L:R) = 2-\sum p_{m_{LR}}S\left(\rho^S_{m,m_{LR}}\right) - H(\{p_{m_{LR}})\},\label{eq: mutual information final exp}\end{align}
where $H(\{p_i\})=-\sum_i p_i\ln p_i$ is the classical entropy. 

Consider $H(\{p_{m_{LR}}\})=0$ which corresponds to when the symmetric charge on $LR$ is fixed, that is the symmetric charge associated with qubits $L,R$ is not getting decohered. This happens, for example, when only one sublattice of the cluster state is being decohered. We immediately see from the above expressions that a non-zero strange correlator for typical trajectories $m$  implies greater than $1$ trajectory mutual information and positive coherent information by eq.\eqref{eq: coherence information in mixed trajectories}. As a special case, for $SC_{L,R}=1$, for example for pure cluster states, the coherent information is maximum. 
As a direct application of the above arguments, we can recover the result in eq.~\eqref{eq: Ic(RM) Z dephaing one sublattice} which was arrived at from the perspective of decoding. It has been computed in \cite{lee2023symmetry} that the Type-\RN{1} strange correlator has the following behavior
\begin{equation}
    \begin{split}
        &SC^{\RN{1}}_{1, 2n+1} = e^{-n/\xi},\ \  \xi = 1/\ln(1/(1-2p)).
    \end{split}
\end{equation}
Putting this in eq.~\eqref{eq: density matrix decomposition} and using eq.~\eqref{eq: mutual information final exp} we find similar decay of $I_c(L:RM)$ as in eq.~\eqref{eq: Ic(RM) Z dephaing one sublattice}.

Similarly, case with $H(\{p_{m_{LR}}\})>0$ can be analyzed and eq.~\eqref{eq: Ic(RM) Z dephasing} can be recovered. Intuitively, higher $H(\{p_{m_{LR}}\})>0$ implies less logical information in the symmetric operator $X_L$ is getting transmitted. Similarly, higher $S(\rho^S_{m,m_{LR}})$ implies less logical information about $Z_L$ being transmitted. The strange correlator on different sublattices is related to the above two entropies and hence to the various behaviors of the coherent information.  

Finally, the observation that $SC_{LR}$ is related to off-diagonal elements of the density matrix is general and not specific to the 1d cluster state. Though the calculation may not remain as simple as above, we expect the results to hold, at least qualitatively, for other SPTs as well.

\subsubsection{Phase transitions in the coherent information and the strange correlator}

Let's consider the 2d cluster state with $Z$-dephasing on the edge qubits. Since the vertex qubits are not being dephased, if we consider information transmission between two far apart vertex qubits, say $v_i$ and $v_j$, logical $\bar{X}$ operator can be transmitted perfectly regardless of the decoherence strength. The transmission of the logical $\bar{Z}$ in each trajectory is diagnosed by the type-\RN{1} strange correlator, which has been shown can undergo a transition as one tune the decoherence strength. Following the calculations in \cite{lee2023symmetry,zhu2023nishimori}, the type-\RN{1} strange correlator is given by
\begin{equation} \label{eq: SC in Z-dephased 2d cluster state}
    \text{SC}^{\RN{1}}_{v_i, v_j}(m) = \left\langle Z_{v_i} Z_{v_j} \right\rangle_{\beta, m_e} 
\end{equation}
The right-hand side of the equation is the correlation function of an Ising model at the inverse temperature $\beta = \tanh^{-1}(1-2p)$ and with bond configuration $m_e$, that is $m_e=-1$ are bonds with antiferromagnetic coupling. For a typical $m_e$, the correlation function would have an order-disorder transition at some finite temperature $p_c$ that depends on the bond configuration $m_e$. Then, some natural questions to ask are whether there is a transition in $I_{c}(L:RM)$, which is described by the trajectory averaged value of mutual information. If there is a transition, then what is its nature, and what is the value of the critical $p_c$?\newline
\indent Before addressing the above questions, we notice that the probability for obtaining measurement outcome $m$ is given by
\begin{equation} \label{eq: RBIM partition function}
    \begin{split}
        p_m \propto Z_{\beta, m_e} = \sum_{\{\sigma\}} e^{\beta \sum_{e = \langle i,j\rangle} m_e \sigma_i \sigma_j},
    \end{split}
\end{equation}
which equals to the partition function of a $RBIM$ given bond configuration $m_e$ and inverse temperature $\beta = \tanh^{-1}(1- 2p)$. Following the discussion in \cite{zhu2023nishimori, lee2022decoding}, if we consider the gauge transformation
    $\sigma_i \rightarrow \tau_i \sigma_i$ and $m_{ij} \rightarrow \tau_i \tau_j m_{ij}$, the partition function can be viewed as a gauge-symmetrized probability distribution of uncorrelated bond disorder with the probability for antiferromagnetic bond given  by $\mathrm{Pr}(m_e=-1)\equiv p_{-} = \frac{1}{1 + e^{2\beta}}$. Such a distribution of bond disorder is said to lie on the Nishimori line\cite{nishimori1981internal}. The RBIM along the Nishimori line has an order to disorder transition at $p_c = 0.109$\cite{honecker2001universality}\newline

\indent Following the intuition that $p_m$ is a gauge-symmetrized bond disorder distribution, the trajectory averaged mutual information in eq.\eqref{eq: coherence information in mixed trajectories} can be written as
\begin{equation}\label{eq: uncorrelated bond disorder averaged coherent information}
\begin{split}
     &\sum_m p_m I^{(m)}(L:R)\\
     &\propto \sum_m \left( \sum_{\sigma} e^{\beta\sum_{\left\langle ij\right\rangle}m_{ij}\sigma_i \sigma_j} \right) I^{(m)}(L:R) \\
     &\propto \sum_{\tau} \left( \sum_{m'} e^{\beta \sum_{\left\langle ij\right\rangle}m_{ij}' } I^{(m')}(L:R) \right) \\
     &\propto [I^{(m)}(L:R)]
\end{split} 
\end{equation}
where from the second to the third line, we gauge fix $m$ to be $m'$ and include a summation over all possible gauge transformation $\tau$. Notice that the mutual information as one can computed from eq.~\eqref{eq: density matrix decomposition} is invariant under this gauge fixing. Therefore, in the fourth line, we show that the trajectory averaged mutual information is equivalent to averaging over uncorrelated bond disorder with the probability for antiferromagnetic bond given by $\mathrm{Pr}(m_e=-1)\equiv p_{-} = \frac{1}{1 + e^{2\beta}}$.\newline
\indent From eq.~\eqref{eq: density matrix decomposition} and~\eqref{eq: SC in Z-dephased 2d cluster state}, the mutual information is a non-linear function of  $|\left\langle Z_{v_i}Z_{v_j}\right\rangle_{m_e, \beta}|$. More precisely, the mutual information $I^{(m)}(L:R)$ is greater than $1$ if the bond configuration $m_e$ is long-range ordered. In the RBIM along the Nishimori line, above the critical inverse temperature $\beta_c$, typical bond configuration $m_e$ is long-range ordered and the fraction of bond configurations that are paramagnetic vanishes in the thermodynamic limit; below the critical inverse temperature $\beta_c$, typical bond configuration $m_e$ is paramagnetic and the fraction of bond configurations that are long-range ordered vanishes in the thermodynamic limit\cite{honecker2001universality, lee2024exact}. From \eqref{eq: density matrix decomposition}, the term $S^{(m)}(\rho_L)$ in each trajectory is simply $1$. As a result, the coherent information $I_c(L:RM)$(or one can say the disorder averaged mutual information $\left[ I^{(m)}(L:R) \right]$) sees the same order-disorder transition in the RBIM along the Nishimori line. When $p < p_c$, $I_c(L:RM) > 0$ and there is quantum information transmitted between $L$ and $R$; when $p > p_c$, $I_c(L:RM) = 0$ and there is only classical information transmitted between $L$ and $R$.   \newline

\section{Mixed SPTs away from the fixed point}\label{sec: away from fixed point}

We can ask if the results we have so far are true for the entire SPT phase, that is, decohering any state from a SPT phase, are the results obtained qualitatively true? The connection between coherent information and type-I and II strange correlators still holds when we move away from the fixed point suggesting that the non-trivial nature of the decohered state should be robust.  

 We give a heuristic argument in 1d that the information transmission is well-defined away from the fixed point. The exact setup we have in mind is as follows. We start with the fixed point of the phase and apply a symmetric perturbation so that the resulting state remains in the SPT phase albeit away from the fixed point. We model the perturbation by symmetric quantum circuits. The decoherence acts on the perturbed state. 

We give a protocol to use the perturbed pure state for transmitting information; we will return to the decohered case later. Let the depth of the circuit be $d$ where the depth is defined as the number of layers in the symmetric circuit (odd and even layers are counted as different layers). The information initially entangled at the boundary qubits are now spread to a distance $d$ away from the boundary. Let us include these extra qubits to define a new boundary $L'=L+\Delta L$ (see the 1st row of Fig.~\ref{fig: non-fixed}). We want to transmit information from $L'$ to $R'$. This can be achieved as follows.

\begin{figure}
    \centering
\includegraphics[width=0.9\linewidth]{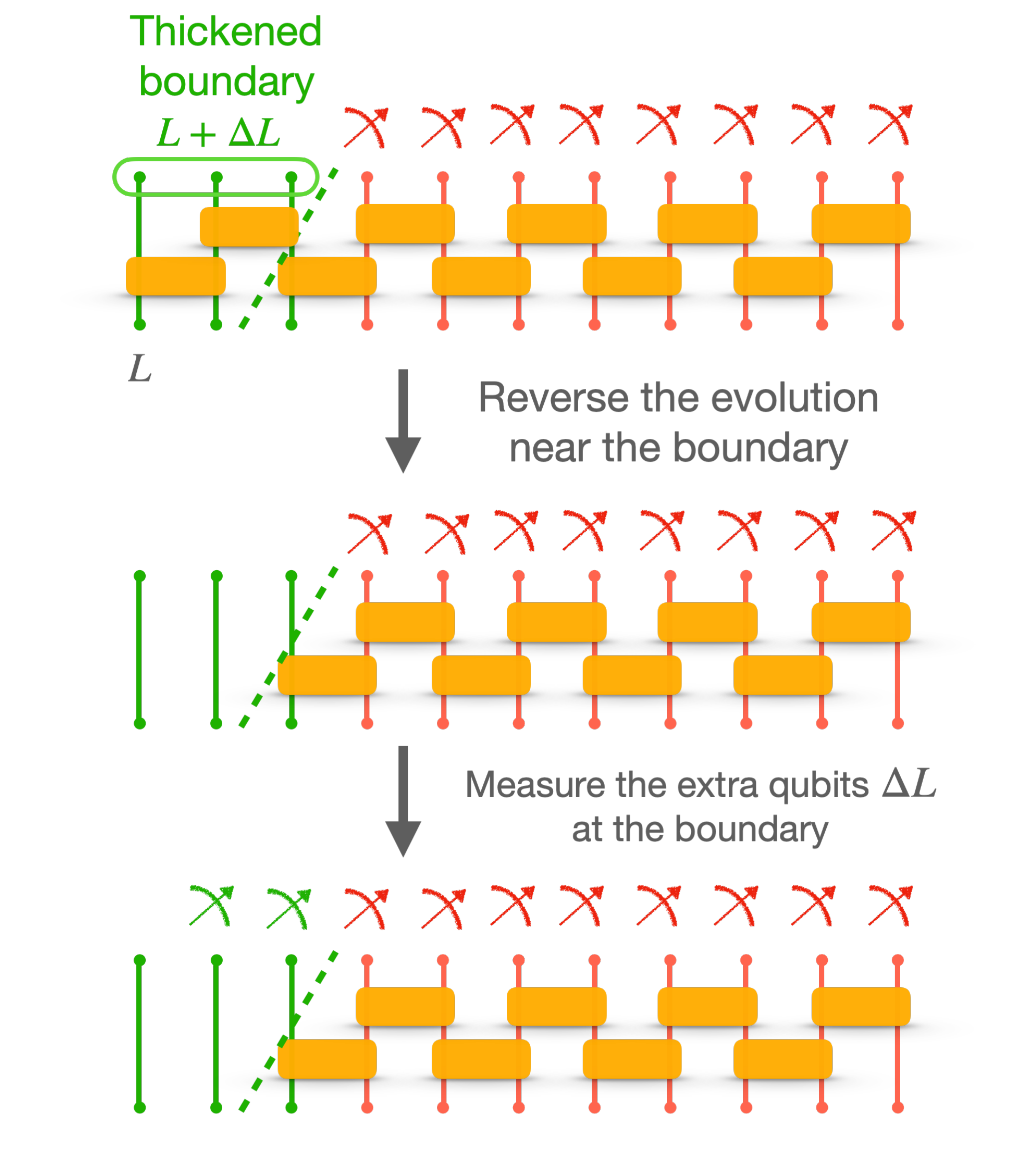}
    \caption{Quantum communication away from the fixed point of pure SPTs. We model the symmetric perturbation by a short-depth symmetric unitary circuit. To transmit information from one edge to another we need to thicken the boundaries to include the qubits inside the light cone (green qubits). The procedure shown in the figure allows for the bulk symmetric charge to be determined from the measurement outcomes (see main text). This leads to the transmission of quantum information. When the circuit depth becomes of the order of the system size the light cone of the boundaries starts to overlap and there is no well-defined notion of quantum communication.}
    \label{fig: non-fixed}
\end{figure}

The boundary $L'$ then applies the inverse of the circuit inside the light cone connecting $L$ and $L'$ (2nd row in Fig.~\ref{fig: non-fixed}). This is always allowed as we are now transmitting information from $L'$ instead of $L$ and any local rotation on $L'$ and $R'$ are allowed. The same thing is done at $R'$. The symmetry charge measurements are now performed everywhere except at $L,R$ and let $\{m_i'\}$ be the measurement outcomes. Since the circuit commutes with the symmetry, the symmetric charge $\prod m_i'$ is also the charge before the circuit is applied. To be more precise, let us perform bulk symmetric charge measurement $\prod_i X_i$ followed by the single site measurements $X_i$. If the short depth circuit is $V$ (excluding the gates near the boundary which got removed when applying the inverse circuit inside the light cones of the boundaries), then the post-measurement state after bulk measurement is $P_{\prod m_i'} V \ket{\psi}=V P_{\prod m_i'}\ket{\psi}$, where $P_{\prod m_i'}$ is projection to $\prod_i X_i = \prod_i m_i'$ ($i$ reside in the bulk). The projection $P_{\prod m_i'}\ket{\psi}$ creates entanglement between $L,R$, and as $V$ does not act on the boundaries, $P_{\prod m_i'} \ket{\psi'}$ also has entanglement across $L,R$. Performing single-site measurements in the bulk won't disturb this entanglement.

Since the value of $\gamma$ labeling the post-measurement logical state depends only on the total symmetric charge, the post-measurement logical state on $L,R$ is the same as that at the fixed point, and thus $I_c^{(0)}(L':R'M)=1$. This is true provided $d$ does not scale with the system size $N$. When $d$ is of order $N$ the two boundaries $L',R'$ boundary overlap. Note that the effect of the unitary circuit may still be present in the non-logical space but it is of no concern for the purpose we have in mind.

Let us now introduce decoherence to the SPT state. After the short-depth symmetric circuit $V$, some qubits are succumbed to decoherence. Let us for simplicity assume that decoherence acts only on one site $i$ in the bulk. Similar to the protocol without decoherence, we reverse the light cone near the boundaries by applying the inverse of the unitary gates. We then measure all bulk qubits in the on-site symmetry basis. The environment can then measure $O^{(m_i)}_i$ depending on the measurement outcome $m_i$ in the bulk. If $U_i$ is the interaction between the system and environment at site $i$, then we have $P^E_{e_i}P^S_{m_i}U_i = P^S_{m_i} P^{SE}_{e_i}U_i = P^S_{m_i}U_i P^S_{e_i}$, where $P^E_{e_i}$ is the projection to eigenvalue $e_i$ of the operator $O^{(m_i)}_i$, $P^S_{m_i}$ is the projection to $X_i=m_i$, and $P^{SE}_{e_i}$ is the projection to $X_i'=U_iX_iU^\dagger_i=e_i$. As argued above, the projector $\prod_i P^S_{e_i}$ does not change under the symmetric circuit. Thus using $\prod_i e_i$ after measuring the environment operators $O^{(m_i)}_i$ we can learn the value of $\gamma$ and determine the bell pair between the boundaries. 

The argument above can be generalized to decoherence acting on multiple qubits outside the light cone of the boundaries ($L+\Delta L, R+\Delta R$). Decoherence inside the light cone can leak some of the information to the environment but $I_c(L':ER'M)$, which also includes the environment qubits entangled within the light cone, should not change. But the decoherence inside the light cone might reduce $I_c(L':R'M)$ and there might be a phase transition with respect to the strength of decoherence. We leave a detailed study of these possibilities for future work.

As a concrete example in two dimensions, we study the effect of decoherence on a perturbed 2d cluster state with Hamiltonian\begin{equation}
    H(\lambda) = -\sum_{e} X_e \prod_{v \in e} Z_v - \sum_v X_v{\prod_{e \ni v} Z_e} + \sum_{v}e^{-\lambda \sum_{e\ni v}X_e}. \label{eq: perturbed 2d cluster state H}
\end{equation}
The strength of the perturbation is controlled by $\lambda$.

We show in Appendix \eqref{sec: Decohered 2d cluster state away from the fixed point} that the ground state wavefunction is given by
\begin{equation}
    \ket{\psi(\lambda)} \propto e^{\lambda \sum_{e} \frac{\hat{X}_e}{2}} \ket{\psi_{0}},
\end{equation}
where $\ket{\psi_{0}}$ is the 2d cluster state fixed-point wavefunction. We calculate type-\RN{1} strange correlator $SC^{\RN{1}}$ and find it to be non-zero in the regime of $1 \gg \lambda > 0$ (where the perturbation is equivalent to adding a small transverse field in the fixed point Hamiltonian). Using results from Sec.~\ref{sec: connection to type I SC}, there thus exists a mixed-state phase for non-zero $\lambda$ where $I_c(L:RM) > 0$. This is in agreement with the heuristic argument above based on quantum circuits.

\section{Discussion and outlook}

\subsection{Mixed SPT as quantum channels}\label{sec: quantum channels/ QEC}
The results in this paper can also be seen from the perspective of quantum error correction. Similar to the pure case~\cite{Raussendorf_2001,Raussendorf_2003,Briegel_2009}, we view decohered SPTs in $d$ dimension as a $(d-1) + 1$ dimensional virtual evolution of the boundary when the bulk is measured in a symmetric basis. In the absence of decoherence, the virtual evolution is unitary.  The decoherence acts as an error in this virtual evolution~\cite{Raussendorf_2005,Roberts_2020,Bolt_2016}. We show in Appendix~\ref{appendix: quantum channel} that having positive $I_c(L:RM)$ is related to having a finite error threshold of this noisy virtual evolution. More precisely, $I_c(L:RM)$ is equal to the amount of information surviving after the noisy evolution. This idea is illustrated for cluster states in various dimensions in Appendix~\ref{appendix: quantum channel}. See also Fig.~\ref{fig:spt as channel}. It is also known that CSS codes can be foliated to cluster state on some graph~\cite{Bolt_2016}. This immediately implies that such cluster states in the presence of decoherence are related to quantum error correction in the corresponding CSS code.

\subsection{Symmetry-decoupling channels and the role of weak symmetry in mixed SPT order}
The coherent information $I_c(L:ERM)$ is non-zero if the channel satisfies \condition condition: the on-site symmetry $G=\otimes_i G_i$ under the system-environment interaction is transformed as $U G_i U^\dagger = \sum_m O^{(m)}_i P_m$, where $P_m$ is projector onto symmetry charge $m$ of $G_i$, and $O^{(m)}_i$ is an $m$ dependent unitary operator acting on the environment qubits. If $G_i$ is measured on the system then a corresponding measurement of $O^{(m_i)}$ can be performed on the environment to learn the local charge at site $i$.
We also present an example of a channel satisfying the above condition but is \emph{not} weakly symmetric. 
We prove in Theorem~\ref{thm 1} that for \condition channels, the measurements do not destroy the quantum information. 

A special and important subclass of these channels is the weakly symmetric channels.
A weakly symmetric channel takes a symmetric pure state $\rho_0$ to $\rho_D = \sum_a K_a\rho_0 K_a^\dagger$ and satisfies,\begin{align}
    G \rho_D G^\dagger = G,
\end{align}
where $G$ is the symmetry and $G\rho_0 = \rho_0 G^\dagger = e^{i\theta_g}\rho_0$. In other words~\cite{de2022symmetry},\begin{align}
    GK_aG^\dagger = \sum_b V_{ab}K_b,
\end{align}
where $V$ is a unitary rotation among the Krauss operators. For \condition channel, we can show that the channel is weakly-symmetric if and only if $mO^{(m)}=)$

Since $\sum_{m_i} m_i^{-1}P_{m_i} = G_i$, weak symmetry condition is equivalent to $UG_iU^\dagger = O_i G_i$, where $O_i$ is some unitary operator on the environment at site $i$. The new symmetry can thus be decomposed as $UGU^\dagger = O\otimes G$, where $O$ is a symmetry action on the environment.

We now ask the question: how important is weak symmetry to have decohered mixed-state SPTs? Or more appropriately, how does decohered symmetry help protect such states against symmetric perturbations? As shown in the text, even a non-weak-symmetric channel preserves the quantum communication property of the pure SPT. This is also stable against decohering a SPT initialized away from the fixed point as shown in Sec.~\ref{sec: away from fixed point}. The existence of mixed-SPT order using the strange correlators also relied on the presence of weak symmetry that results in symmetry in the doubled Hilbert space of the density matrix. However, as shown in this work, the strange correlator even without weak symmetry can be non-vanishing. These and other related questions suggest a more careful study of the role of weak symmetry in protecting the mixed SPT order is required. 

One consequence of having weak symmetry is as follows. A given density matrix can be decomposed in a non-unique was as $\rho=\sum_a p_a \ket{\psi_a}\bra{\psi_a}.$ The trajectories $\ket{\psi_a}$ can be thought of as performing measurements on the environment on a specific basis. More precisely, there exists a purification, $\ket{\Psi_{SE}} = \sum_a \sqrt{p_a} \ket{\psi_a}\ket{a_E}$ such that measuring the environment in basis $\{\ket{a_E}\}$ projects the system to trajectories $\{\ket{\psi_a}\}$. For mixed SPTs with $I_c(L:ERM)>0$ and weak symmetry with composite symmetry $O\otimes G$, we can project the environment to eigenstates of $O$ to get symmetric non-trivial states, that is, the system's trajectories can be used to transmit quantum information by measuring the bulk. In other words, there is a decomposition $\rho=\sum_a p_a \ket{\psi_a}\bra{\psi_a}$ such that each of the (or typical) $\ket{\psi_a}$ has non-trivial SPT order or edge modes for weakly-symmetric mixed states and $I_c(L:ERM)>0$. 
This also means there is an ensemble of Hamiltonians $\{H_a\}$ whose ground states $|\psi_a\rangle$ are resources for quantum communication. Thus $I_c(L:ERM)>0$ and weak symmetry implies the existence of a disordered Hamiltonian with average SPT~\cite{ma2023average}.

\subsection{Connections to other probes of mixed SPT order}
The quantum communication ability of mixed SPTs can also be related to other probes. As shown in Sec.~\ref{sec: connection to SC}, the coherent information has intimate connections to strange correlators~\cite{lee2023symmetry,zhang2022strange}. There we prove that \condition channels imply both $I_c(L:ERM)$ and type-II strange correlator to be non-zero, suggesting a connection between them. In \cite{zhang2022strange}, the authors considered a strange correlator defined using fidelity between the mixed-SPT density matrix and a trivial state with weak symmetry. Such a trivial state can be thought as an ensemble of pure states with different symmetry charges. Therefore, such a fidelity strange correlator also has the spirit of summing over all ``measurement outcomes of symmetry charges" as $I_c(L:ERM)$. We believe the symmetry-decoupling channels would also imply a non-trivial fidelity strange correlator, and we leave the rigorous proof to future works. Moreover, positive $I_c(L:RM)$ implies that all type-I strange correlators are non-zero. On the other hand, a zero value for $I_c(L:RM)$ is an indication of the presence of some long-range classical correlation and some of the strange correlators might be zero.

The classification of mixed SPTs based on separability  introduced in~\cite{chen2023symmetry, chen2023separability} relies on the existence of a decomposition $\rho=\sum_a p_a\ket{\psi_a}\bra{\psi_a}$ such that each $\ket{\psi_a}$ is trivial. A density matrix is called symmetric long-range entangled if such a decomposition does not exist. As discussed above, $I_c(L:ERM)$ is a diagnosis for the existence of decomposition where each trajectory is symmetric long-range entangled. Thus there is no direct connection between these two probes though we believe $I_c(L:RM)$ to be closely connected to separability criteria. When $I_c(L:RM)$ is maximum the density matrix is symmetric long-range entangled based on separability since, irrespective of the measurement outcomes in the environment, the system's mixed state is a resource to transmit quantum information. Thus every decomposition of such a density matrix will have trajectories with non-trivial edge modes. But in addition to this, we also find for 1d, 2d cluster states that for $I_c(L:RM)=0$ (the classical information can be transmitted), the density matrix is symmetric long-range entangled using separability as shown in~\cite{chen2023symmetry}. This motivates the conjecture that a density matrix has symmetry-protected long-range entanglement based on separability if the coherent information $I_c(L:RM)\geq 0$ and vice versa. To what extent $I_c(L:RM)$ and the separability probe are connected is left for future work. 

Another approach to defining mixed state order is to use an equivalency class of mixed states under finite depth local channels~\cite{sang2023mixed,ma2023average, sang2024stability}. For SPTs, a mixed state is considered trivial if it can be prepared or made trivial using a symmetric finite depth local channel. We leave the connection of the behavior of coherent information with that of the equivalency class to future work. However, one thing is clear any mixed state capable of communicating quantum information should not be able to be prepared starting from a trivial state in finite quantum time. This is so because the channel can be purified using an environment and the combined system and environment should be trivial if the system was in a trivial state to begin with. 

In this work, we focussed on examples with the boundary logical space of size independent of the system size. Generally, higher-forms symmetries should not be able to transmit extensive amounts of information and one needs to consider subsystem symmetry-protected topological states, SSPTs~\cite{you2018subsystem}, such as states with line symmetry~\cite{raussendorf2019computationally} and fractal symmetry~\cite{devakul2018universal}, to get a SPT phase capable of communicating extensive amount of information. The formalism introduced in this paper can be easily extended to SSPTs. We leave this open for future work.

\textit{Note added.} While writing this manuscript two preprints appeared on arxiv~\cite{sala2024spontaneous,lessa2024strongtoweak}. The authors studied spontaneous strong symmetry breaking (SSSB) for mixed states. The studies suggest that SSSB can be understood as the system being in a SPT state with the environment. We leave a detailed analysis between SSSB and mixed SPTs as defined in this paper for future consideration. 

\acknowledgements{Z.Z, U.A, and S.V. thank Yimu Bao, Yu-Hsueh Chen, Tarun Grover, and Ali Lavasani for helpful discussion. Z.Z also thanks Tim Hsieh, Tsung-Cheng Lu, Yichen Xu, and Jian-Hao Zhang for helpful discussions. This work was supported by the Simons Collaboration on Ultra-Quantum Matter, which is a grant from the Simons Foundation (651440, U.A.), and an Alfred P. Sloan Research Fellowship (S.V.)}

\appendix

\section{Coherent information in trajectories} \label{app: i_c}

Consider a quantum channel $\mathcal{N}$ which acts on a initial density matrix $\rho_0$ as follows,\begin{align}
    \mathcal{N}[\rho_0] = \sum_m p_m \mathcal{N}_m[\rho_0] \otimes \ket{m}\bra{m},
\end{align}
where $m$ is the trajectory label and may, for example, correspond to measurement outcomes, $\mathcal{N}_m$ is the channel in trajectory $m$. We say that the channel has pure trajectories if $\mathcal{N}_m$ maps pure state to pure state. 

We want to study the coherent information of such channels.
To do so, we take the initial state to be a mixed state with one bit of entropy or entanglement. We can think of it as a pure initial state that is entangled with an ancilla $A$ and the combined state of $A$ and the system is pure initially. Let the state of the system at a later time be denoted by $Q$. Then the coherent information about $A$ in $Q$ and $M$ is given by\begin{align}
    I_c = S(\rho_{QM}) - S(\rho_{AQM}).
\end{align}

\paragraph{Channel with pure trajectories.} For pure trajectories\begin{align}
    \rho_{AQM} = \sum_m p_m \ket{\psi^{(m)}_{AQ}}\bra{\psi^{(m)}_{AQ}} \otimes \ket{m}\bra{m}.
\end{align}
Putting this in the expression for the coherent information we get\begin{align}
    I_c &= \sum_m p_m \left(  S(\rho^{(m)}_{Q}) - \log p_m \right) - \sum_m -p_m \log(p_m)  \nonumber \\
        &= \sum_m p_m S(\rho_Q^{(m)}).
\end{align}
This is the usual result for the coherent information for pure trajectories.

\paragraph{Mixed trajectories.} For channel with mixed trajectories we instead have \begin{align}
    \rho_{AQM} = \sum_m p_m \rho_{AQ}^{(m)} \otimes \ket{m}\bra{m},
\end{align}
where $\rho_{AQ}^{(m)}$ is no longer pure. In this case the coherent information is given by,\begin{align}
    I_c &= \sum_m p_m \left(  S(\rho^{(m)}_{Q}) - \log p_m \right) - \sum_m p_m \left(  S(\rho^{(m)}_{AQ}) - \log p_m \right) \nonumber \\
        &= \sum_m p_m  \left(  S(\rho^{(m)}_{Q}) - S(\rho^{(m)}_{AQ}) \right) \nonumber \\
        &= \sum_m p_m I^{(m)}(A:Q) - \sum_m p_m S(\rho_A^{(m)}) \nonumber \\
        &= \sum_m p_m I^{(m)}(A:Q) - (|A| - I_{\mathrm{dest}}) ,
\end{align}
where $I^{(m)}(A:Q)$ is the mutual information between A and Q for trajectory $m$, and $I_{\mathrm{dest}}$ is the amount of coherent information destroyed by the measurements.

\section{Calculations for decoder}\label{app: decoder}
Let the SPT state in 2 dimension $\rho_0$ be decohered to $\rho = \sum_i K_i \rho_0 K_i^\dagger$ where $K_i$ are Krauss operators. A measurement performed on qubit $i$ is termed erroneous if the actual symmetric charge at that site is different than what is observed. Such locations are called error locations. The probability of a qubit being an error location given measurement outcomes $m$ is given by,\begin{align}
    P(\epsilon|m) = \mathrm{Tr} \prod_i \frac{1-m_iG_i'}{2} \rho_m,
\end{align} 
where $\rho_m = P_m \rho_{SE} P_m / \mathrm{Tr}(P_m\rho_{SE})$ is the post-measurement state, $G_i'$ is the modified symmetry charge operator supported on the system and the environment. 

For Z-dephasing with strength $p$, $G_i' = Z^E_i X_i$ and $P(\epsilon,m) = \prod_{i\in \epsilon} p \prod_{i\notin \epsilon}(1-p) \times \sum_{C} \prod_{i\in C} m_i {\epsilon_i}$, where $C$ is a close loop on the lattice. 

Similar calculations can be done for \condition channels. The channel described in eq.~\eqref{eq: SDC U} is acted on the system with probability $q$ and with probability $1-q$ no decoherence happens. We want to calculate the probability $P(m,\epsilon)$ where $\epsilon$ is a set of qubits with erroneous observed outcomes. It can be calculated as follows
\begin{widetext}
\begin{align}
    P(m,\epsilon) &= \mathrm{Tr} \prod_{i\in \epsilon} \frac{1-mG_i'}{2}  \prod_{i\notin \epsilon}\frac{1+mG_i'}{2} \left(\frac{1+mX_i}{2} \rho_{SE} \frac{1+mX_i}{2}\right),
\end{align}
\end{widetext}
where $G_i'$ is the modified local symmetric charge. The above expression can be simplified to \begin{align}
    P(m,\epsilon) \sim &  \prod_{i\in \epsilon} q_i \prod_{i \notin \epsilon} \left(1-q + q_i\right)\ \  \times  \label{eq: disorder distribution}\\
    & \times \sum_{C} \prod_{i\in C} m_i {\epsilon_i}, \nonumber
\end{align}
where $C$ is a close loop on the lattice, $q_i =q \left(\frac{1+rm_i}{2}\right) $, $r=\langle X^E \rangle$ is the average value of the environment's qubits initial state, and $\epsilon_i = -1$ if $i$ is an erroneous qubit and $+1$ otherwise. Here we abuse the notation by denoting the set of erroneous qubits by $\epsilon$ and $\epsilon_i = \pm 1$ also being a variable defined on the erroneous set. Finally, the last term $\prod_{i\in C} m_i {\epsilon_i}$ is a constrain on the error path given measurement outcomes (or the other way round), that is, the error paths can only begin and end on plaquettes with $\prod_{i\in \square}m_i = 1$. This is a generic constraint not special to the model under study. We do away with this term by putting the above constraint in the summation of the error paths $\epsilon$. This also means that summation over $\epsilon$ is performed before summation over $m_i$.

To transmit the information across the cluster state one needs to know the symmetric charge string running across the boundary. Without decoherence, the string is just the product of the measurement outcomes along any string. However, with decoherence, there are bonds with erroneous measurement outcomes. In this case, a decoder is needed to learn the correct symmetric string. The probability that the decoder successfully learns the charge string is equal to the probability that the decoder guesses an error path that is equal to the actual error path up to a close loop, $\epsilon' = \epsilon + \omega$ where $\omega$ is a close loop on the dual lattice. A close loop of error is benign as any charge string will pass an even number of times through the error loop and thus would not change the true value of the charge string. There are different equivalency classes of error paths which are distinguished by the presence or absence of non-contractible loops. If the decoder mistakenly adds a non-contractible loop then the prediction from the decoder would be wrong. The probability of the error lying in an equivalency class is given by,\begin{align*}
    P(m,\bar{\epsilon}) = \sum_{\omega} P(m,\epsilon+\omega), 
\end{align*}    
where $\omega$ are close loops. Only endpoints of the erroneous strings $\epsilon$ can be known a priori by recognizing plaquettes with $\prod_{i\in \square}m_i=-1$. Let us denote the endpoints by $s$. The input to the decoder are endpoints of the error path $s$ and it chooses a random path $\epsilon$ connecting the endpoints. The decoder then computes $P(m,\overline{\epsilon+\lambda_i})$, where $\lambda_i$ are non-contractible loops distinguishing different equivalency classes. The output of the decoder is the equivalency class with maximum probability, or likelihood,\begin{align}
    \delta = \mathrm{argmax}_i P(m,\overline{\epsilon+\lambda_i}) 
\end{align} 
Below the threshold when the decoder is supposed to work we expect that $P(m,\overline{\epsilon+\delta})\rightarrow 1$. An useful diagnostic for the decoder performance is $\Delta(m,s) = \log \frac{P(m,\overline{\epsilon+\delta})}{P(m,\overline{\epsilon+\delta'})}$,
where $\delta'$ is the less likely equivalency class and goes from $\infty$ to $0$ across the threshold transition.

Using standard mapping of the error correcting codes to stat mech models, the probability of the equivalency is proportional to a partition function of a quenched RBIM coupled to another RBIM at zero temperature,\begin{align}
    P(m,\overline{\epsilon}) \sim \left(\sum_{\tau} e^{\beta \sum_{ab} J_{ab} \epsilon_{ab} \tau_a\tau_b + hm_{ab}}\right), \label{eq: stat mech model}
\end{align}
where $\tau_a$ are Ising spins living on the plaquettes of the original lattice. We have \begin{align*}
    e^{\beta J_{ab} + hm_{ab}} &\propto 1-q+q_{ab} \\
    e^{-2\beta J_{ab}} &= q_{ab}/(1-q+q_{ab}).
\end{align*}
Note that $q_{ab}$ and hence $J_{ab}$ depends on $m_{ab}$. Also, we denote the same bond in two different ways, $(ab)$ and $i$, where the former indices live on the dual lattice and the latter index lives on the original lattice. 
The average performance of the decoder is probed by\begin{align}
    \langle \Delta \rangle = \sum_m  \sum_{\epsilon} P(m,\epsilon) \Delta(m,\epsilon).
\end{align}
The above quantity is a quenched average of the stat-mech model described in eq.~\eqref{eq: stat mech model} with quench disorder variables $m_{ab},\epsilon_{ab}$ distributed according to distribution $P(m,\epsilon)$ in eq.~\eqref{eq: disorder distribution}.
Though we are not able to solve the above statistical mechanic model, we believe it to have an order-disorder transition with respect to $q$. We leave a detailed study of the transition for another work.

For Z-dephasing however, exact results can be known. The probability of an equivalency class in this case is given by, \begin{align}
P(m,\bar{\epsilon})&=\sum_{\omega}P(m,\epsilon+\omega) \\
    &\sim \sum_{\sigma} e^{\beta\sum_{ij}\epsilon_{ij}\sigma_i\sigma_j},
\end{align}
where $\sigma_i$ are Ising spins, and $e^{-2\beta}=p/(1-p)$. The resulting model is the 2d RBIM at the Nishimori line which has been extensively studied in the context of decoder for Toric Code~\cite{dennis2002topological}.

\section{Properties of \condition channels}
\subsection{Symmetry properties of \condition channels} \label{app: SDC under symmetry}

The \condition channel (SDC) is generated by system-environment on-site interaction\begin{align}
U_i=CNOT\cdot e^{i\theta Y^E_i}\cdot SWAP. 
\end{align}
The environment qubits are initialized in the product state of $\ket{e_0}$. The original symmetry $\prod X$ is transformed to $X_i' = UX_iU^\dagger = aX_iX_i^E + bIZ_i^E$, where $a=\sin \theta$ and $b=\cos \theta$. Let us for simplicity of notation consider SDC acting only on one site, $i$. The density matrix of the system evolves as \begin{align}
    \rho = \sum_{\alpha=0,1} K_\alpha \ket{\psi} \bra{\psi} K_\alpha^\dagger,
\end{align}
where $\ket{\psi}$ is the initial state of the system and is symmetric. The Krauss operators $K_i$ are defined as \begin{align}
    K_\alpha \ket{\psi} = \ket{\alpha^E} U_i \ket{\psi e_0},
\end{align}
where $\ket{\alpha^E}$ are eigenstates of $Z_i^E$ (or any other complete basis).
We note that \begin{align}
    \prod X K_0 \ket{\psi} &= \bra{0^E} \prod X  U_i \ket{\psi e_0} \nonumber \\
    & = \bra{0^E} \prod X  U_i \prod X \ket{\psi e_0} \nonumber \\
    & = \bra{0^E} X_i X_i'  U_i  \ket{\psi e_0} \nonumber \\
    &= a K_1 + b X_i K_0.
\end{align}
Similarly, \begin{align}
      \prod X K_1 = aK_0 - b X_i K_1.
\end{align}

We also observe that if the environment starts in the eigenstate of $X$ then $X_i K_\alpha = K_\alpha$. This is due to the relation $X_i U = U X_i^E$.

The channel is called weak-symmetric if $\prod X \rho \prod X=\rho$. Calculating the left-hand side for SDC using the above equations leads to\begin{align}
    \prod X \rho \prod X &= |a|^2 \rho + |b|^2 X_i \rho X_i + \nonumber \\
    &+ \left(ab K_1\rho_0 K_0^\dagger X_i - ab K_0 \rho_0 K_1^\dagger X_i\right) + \mathrm{h.c}
\end{align}  
and is clearly not equal to $\rho$ for generic $\rho$. However, there are special cases where the channel is weakly symmetric. The first case is when $b=0$. The second, more subtle, case is for when the environment is initialized in eigenstates of $X$ operator. Using $X_i K_\alpha = K_\alpha$ we can show that the weak-symmetry condition is satisfied.

Another way to see non-weak-symmetry of the channel is using explicit calculation of $K_\alpha$. Again restricting to single-site decoherence we have the following. If the environment initial state $\ket{e_0}$ is generally given by $\ket{e_0} = \cos\phi \ket{0^E} + \sin\phi \ket{1^E}$, $K_0$ and $K_1$ are given by
\begin{equation}
    \begin{split}
        &K_0  \equiv \bra{0^E} U  \ket{e_0} \\
        &= \frac{1}{2}\cos(\theta - \phi) + \frac{1}{2}\sin(\theta + \phi) X + \frac{1}{2}\sin(\theta - \phi)iY\\
        &+ \frac{1}{2}\cos(\theta + \phi) Z   \\
        &K_1 \equiv \bra{1^E} U \ket{e_0} \\
        &= \frac{1}{2}\cos(\theta + \phi) - \frac{1}{2}\sin(\theta - \phi) X + \frac{1}{2}\sin(\theta + \phi)iY\\
        &- \frac{1}{2}\cos(\theta - \phi) Z   \\
    \end{split}
\end{equation}
where $\alpha \equiv \cos\theta + \sin\theta$ and $\beta \equiv \cos\theta - \sin\theta$. \newline
\indent To determine whether the channel is weakly-symmetric or not under $\prod X$, we want to check if there exists a unitary $x$ such that $\left(\prod X\right) K_i \left(\prod X\right) = \sum_{j} x_{ij} K_j$\cite{de2022symmetry}. That is, after the symmetry transformation, the Kraus operators are rotated by a unitary. We then find that, in addition to $\cos\theta=0$, for $\phi = \frac{\pi n}{2} - \frac{\pi}{4}$, where $n$ is an integer, the channel is weakly-symmetric. 

\subsection{Type-\RN{2} Strange correlator}\label{app: explicit SC calculation under SDC}

In this section, we want to present a detailed calculation of the type-\RN{2} strange correlator for the 1d cluster state decohered by a symmetry-decoupling channel. A local symmetry-decoupling channel  can be written as
\begin{equation}
    \begin{split}
        &\mathcal{E}_{i}[\rho] = (1 - p) \rho + p K_0\rho K_0^{\dagger} + p K_1 \rho K_1^{\dagger},
    \end{split}
\end{equation}
where $K_0$ and $K_1$ are given by the expressions above.

\indent To facilitate our calculation of strange correlators, we first want to find what is the conjugate channel acting on the density matrix $\frac{1 + X}{2}$,
\begin{equation}
    \begin{split}
        \mathcal{E}^{*}[\frac{1 + X}{2}] &= (1 - p) \frac{1 + X}{2} + \frac{p}{2} + \frac{p}{2} (K_0^{\dagger} X K_0 + K_1^{\dagger} X K_1) \\
        &= (1- p) \frac{1 + X}{2} + \frac{p}{2} +  \frac{p}{2}\sin(2\phi) \\
        &= \frac{1 + p\sin(2\phi) + (1 - p) X}{2}.
    \end{split}
\end{equation}
The numerator of the type-\RN{2} strange correlator on odd sublattice is given by
\begin{equation}
    \begin{split}
        &\tr[\mathcal{E}[\rho]Z_{2i + 1} Z_{2j + 1} \rho_0 Z_{2j + 1}Z_{2i + 1}] \\
        &= \tr[\rho Z_{2i + 1} Z_{2j + 1} \mathcal{E}^*[\rho_0] Z_{2j + 1} Z_{2i + 1}] \\
        &= \bra{\psi} \left( \prod_{k = 2i + 1}^{j} X_{2k} \right) \mathcal{E}^*[\rho_0] \left( \prod_{l = 2i + 1}^{j} X_{2l} \right) \ket{\psi} \\
        &= \tr[\rho \mathcal{E}^*[\rho_0]] \\
        &= \tr[\mathcal{E}[\rho] \rho_0]
    \end{split}
\end{equation}
since $\mathcal{E}^*[\rho_0]$ only contains identities and Pauli-$X$s, $\prod_{k=i+1}^{j}X_{2k}$ acting simultaneously on both sides of $\mathcal{E}^*[\rho_0]$ are cancelled. The rest of the term is nothing but the denominator of the type-\RN{2} strange correlator. Therefore 
\begin{equation}
    \frac{\tr[\mathcal{E}[\rho]Z_{2i +1} Z_{2j + 1}\rho_0 Z_{2j + 1}Z_{2i+1}]}{\tr[\mathcal{E}[\rho]\rho_0]} =1.
\end{equation}
\indent The numerator of the type-\RN{2} strange correlator on even sublattice is given by
\begin{equation}
    \begin{split}
        &\tr[\mathcal{E}[\rho] Z_{2i}Z_{2j} \rho_0 Z_{2j}Z_{2i}] \\
        &= \tr[\mathcal{E}[\rho] \left( \prod_{k \neq 2i, 2j}\frac{1 + X_k}{2} \right) \frac{1 - X_{2i}}{2} \frac{1 - X_{2j}}{2}] \\
        &= \bra{\psi} \left( \prod_{k} \frac{1 + X_{2k + 1}}{2} \right) \left( \prod_{l \neq i,j} \frac{1 + p\sin(2\phi) + (1 - p) X_{2l}}{2} \right)\\
        &\frac{1 - p\sin(2\phi) - (1 - p)X_{2i}}{2} \frac{1 - p\sin(2\phi) - (1 - p)X_{2j}}{2} \ket{\psi} \\
        &= 2\frac{\left(1 + p\sin(2\phi)\right)^{L - 2} \left( 1 - p\sin(2\phi) \right)^2 }{2^{2L}} + 2\frac{\left( 1-p \right)^{L}}{2^{2L}},
    \end{split}
\end{equation}
where to get the last line we notice that only operators $\mathds{1}$, $\prod_i X_{2i+1}$, $\prod_i X_{2i}$, and $\prod_{i} X_i$ have non-zero expectation value.\newline
\indent Similarly, we get the denominator to be
\begin{equation}
    \begin{split}
        \tr[\mathcal{E}[\rho] \rho_0] = 2\frac{(1 + p\sin(2\phi))^L}{2^{2L}} + 2\frac{(1 - p)^L}{2^{2L}}
    \end{split}
\end{equation}
The type-\RN{2} strange correlator on even sublattice equals to
\begin{equation}
\begin{split}
    &\frac{\tr[\mathcal{E}[\rho]Z_{2i} Z_{2j}\rho_0 Z_{2j}Z_{2i}]}{\tr[\mathcal{E}[\rho]\rho_0]} \\
    &= \frac{\left(1 + p\sin(2\phi)\right)^{L - 2} \left( 1 - p\sin(2\phi) \right)^2 + \left( 1-p \right)^{L}}{(1 + p\sin(2\phi))^L + (1 - p)^L}.
\end{split}
\end{equation}
Therefore, we see the type-\RN{2} strange-correlator on even sublattice exhibits an area law, which is independent of the separation between site $2i$ and $2j$.

\section{Dephased 2d cluster state} \label{sec: detail calculation of Z-dephased 2d cluster state}
For completeness of the paper, we present the detailed calculations in Sec~.\ref{sec: connection to type I SC} and similar calculations are did in\cite{lee2023symmetry,zhu2023nishimori, lee2022decoding}. 

\subsection{Strange correlator in typical trajectories}
As was explained in the main context, to facilitate the calculation of mutual information in the post-measurement density matrix $\rho^{D}_m$, we want to compute the type-\RN{1} strange correlator $\text{SC}^{I}$ with respect to the trivial density matrix $\rho_{m}$. \newline
\indent A trick to compute the strange correlator is to consider the conjugate channel acting on $\rho_m$ and the resulting operator is given by
\begin{equation}
    \begin{split}
        \mathcal{E}^*[\rho_m] &= \prod_v \frac{1 + m_v \hat{X}_v}{2} \prod_e \mathcal{E}^{*}_e[\frac{1 + m_e \hat{X}_e}{2}] \\
        &= \prod_v \frac{1 + m_v \hat{X}_v}{2} \prod_e \left[ (1 - p)\frac{1 + m_e \hat{X}_e}{2} + \right.\\
        & ~~~~~~~~~~~~~~~~~~~~~~~~~~~~~ \left.+ p\hat{Z}_e\frac{1 + m_e \hat{X}_e}{2}\hat{Z}_e \right] \\
        &= \prod_v \frac{1 + m_v \hat{X}_v}{2} \prod_e \frac{1 + (1 - 2p) m_e \hat{X}_e}{2}.
    \end{split}
\end{equation}
The denominator can then be computed as
\begin{equation}
    \begin{split}
        &\tr[\mathcal{E}[\rho]\rho_m] \\
        &= \tr[\rho \mathcal{E}^*[\rho_m]]\\
        &= \bra{\psi_0} \prod_v \frac{1 + m_v \hat{X}_v}{2} \prod_e \frac{1 + (1 - 2p)m_e \hat{X}_v}{2} \ket{\psi_0}\\
        &= \frac{2}{2^{N_v - 1}} \sum_{\gamma} \prod_{e \in \gamma} (1-2p)m_e
    \end{split}
\end{equation}
The only terms in the product $\prod_v \frac{1 + m_v \hat{X}_v}{2}$ that would contribute are $\frac{1}{2^{N_v}}$ and $\frac{1}{2^{N_v}} \prod_v \hat{X}_v$. In the product over all edges $e$, only edges that form a closed loop $\gamma$ on the dual lattice would contribute. Moreover, by defining dual Ising spins $\sigma$ of the center of each plaquette, we can rewrite the summation over $\gamma$ as the partition function of a RBIM whose bonds are specified by $m_e$ at inverse temperature $\beta = \tanh^{-1}(1- 2p)$,
\begin{equation}
    p_m = \tr[\mathcal{E}[\rho]\rho_m] \propto \sum_{\sigma} e^{\beta \sum_{\left\langle ij \right\rangle}m_{ij}\sigma_i \sigma_j}.
\end{equation}
We also notice that the above partition function can be interpreted as the gauge-symmetrized probability of uncorrelated bond disorder $\{m_e\}$ with $p(m_e = -1) = \frac{1}{1 + e^{2\beta}}$. The gauge transformation is given by
\begin{equation}
    \sigma_i \rightarrow \tau_i \sigma_i,~~~ m_{ij} \rightarrow \tau_i \tau_j m_{ij}.
\end{equation}

\subsection{Averaging $I^{(m)}(L:R)$ over trajectories}
Here we want to present the detailed calculation of mapping the trajectory averaged mutual information to averaging the mutual information by uncorrelated bond distribution.
\begin{equation}
    \begin{split}
        &\sum_m p_m I_{m}(L:R) \\
        &\propto \sum_m \left(\sum_{\sigma} e^{\beta\sum_{\left\langle ij\right\rangle}m_{ij} \sigma_i \sigma_j} \right) I_m(L:R) \\
        &\propto \sum_{\tau} \left( \sum_{m'} e^{\beta \sum_{\left\langle ij \right\rangle} m_{ij}'} I_{m'}(L:R) \right),
    \end{split}
\end{equation}
where from the second line to the third line we gauge-fixing $m$ to be $m'$ and summing over all possible gauge transformations as $\tau$. Therefore, the mutual information weighted over the gauge-symmetrized probability distribution is proportional to the mutual information weighted over the uncorrelated probability distribution.

\subsection{Critical behavior of $I^{(m)}(L:R)$}
 In this section, we want to explicitly show that at the criticla point $\beta_c$ the mutual information weighted by the uncorrelated bond distribution can be mapped to $\sigma_{v_i}\sigma_{v_j}$ correlation function in the quenched disorder RBIM.
\begin{equation}
    \begin{split}
        &\sum_m p_m (\text{SC}_m^{\RN{1}})^2 \\
        &= \sum_m \left(\sum_{\sigma'} e^{\beta \sum_{ij}m_{ij}\sigma_i' \sigma_j'}  \right) \left\langle \sigma_{i}\sigma_{j} \right\rangle_{m, \beta}^2 \\
        &= \sum_{\tau} (\tau_i \tau_j)^2 \left( \sum_{m'} e^{\beta \sum_{ij}m'} \left\langle \sigma_i \sigma_j \right\rangle_{m', \beta}^2 \right) \\
        &= [\left\langle \sigma_i \sigma_j \right\rangle^2]_{\beta}
    \end{split}
\end{equation}
where $[\cdot]$ denotes the quenched disorder average over uncorrelated bond configurations with $p(m_e = -1) = \frac{1}{1 + e^{2\beta}}$. \newline
\indent It has already been shown that along the Nishimori line $[\left\langle \sigma_i \sigma_j \right\rangle^2]_{\beta} = [\left\langle \sigma_i \sigma_j \right\rangle]_{\beta}$\cite{nishimori1981internal}.

\section{Decohered 2d cluster state away from the fixed point} \label{sec: Decohered 2d cluster state away from the fixed point}
In this section, we consider the 2d cluster state in eq.\eqref{eq: Hamiltonian of the 2d cluster state} under symmetric perturbation and study $I_{c}(L:RM)$ under $Z$-decoherence to edge qubits. The perturbed Hamiltonian reads
\begin{equation} \label{eq: Nonlinearly perturbed 2d cluster state}
    H(\lambda) = -\sum_{e} X_e \prod_{v \in e} Z_v - \sum_v X_v{\prod_{e \ni v} Z_e} + \sum_{v}e^{-\lambda \sum_{e\ni v}X_e},
\end{equation}
where the last term is a nonlinear symmetric perturbation to the fixed point 2d cluster state Hamiltonian. The phase diagram of the model will be clear upon doing a generalized Kennedy-Tasaki(KT) transformation~\cite{li2023non} which transforms the operators as,
\begin{equation}
\begin{split}
    X_e \to X_e&, ~~~X_v \to X_v \\
    X_e \prod_{v \in e} Z_v \to \prod_{v \in e} Z_v&, ~~~ X_v \prod_{e \ni v} Z_e \to \prod_{e \ni v} Z_e
\end{split}
\end{equation}
We then get the Hamiltonian of the dual model, which reads
\begin{equation}
    H_{\text{dual}} = -\sum_{e} \prod_{v \in e} Z_v - \sum_v {\prod_{e \ni v} Z_e} + \sum_{v}e^{-\lambda \sum_{e\ni v}X_e}.
\end{equation}
After taking account into the $\mathbb{Z}_2^{(0)}$ 0-form symmetry $\prod_{v} X_v = 1$ and the $\mathbb{Z}_2^{(0)}$ 1-form symmetry $\prod_{e \in \gamma} X_e = 1$ (for any closed loop $\gamma$ defined on the dual lattice), $H_{\text{dual}}$ describes decoupled GHZ state on the vertices and perturbed Toric Code on the edges. Such a perturbed Toric Code model has been studied in~\cite{castelnovo2008quantum} and its exact groundstate wavefunction is given by
\begin{equation}
\begin{split}
    \ket{\psi_{\text{dual}, e}(\lambda)} \propto e^{\lambda \sum_{e} \frac{\hat{X}_e}{2}} \ket{\psi_{\text{TC}}},
\end{split}
\end{equation}
where $\ket{\psi_{\text{TC}}}$ is the fixed-point Toric code wavefunction on the edges. If we represent the Toric code wavefunction as equal weight superposition of loop configurations, the operator $e^{\lambda \sum_e \frac{\hat{X}_e}{2}}$ can be thought of imposing loop tensions to them. It is known that such a perturbed TC model undergoes a transition from topologically ordered phase to trivial phase at $\lambda = 1$, and the cirtical point is in the 2d Ising universality class.

We can infer the phase diagram of the original model in eq.\eqref{eq: Nonlinearly perturbed 2d cluster state} from that of the dual model. The complete groundstate wavefunction of the dual model reads
\begin{equation}
\begin{split}
    \ket{\psi_{\text{dual}}(\lambda)} \propto e^{\lambda \sum_{e} \frac{\hat{X}_e}{2}} \ket{\text{GHZ}} \otimes \ket{\psi_{\text{TC}}},
\end{split}
\end{equation}
where $\ket{\text{GHZ}}$ is the GHZ-state wavefunction on the vertices.
Since the term $e^{\lambda \sum_{e} \frac{\hat{X}_e}{2}}$ remains unchanged under the KT transformation, the groundstate wavefunction of the original model is simply given by
\begin{equation} \label{eq: perturbed 2d cluster state wf}
    \ket{\psi(\lambda)} \propto e^{\lambda \sum_{e} \frac{\hat{X}_e}{2}} \ket{\psi_{0}},
\end{equation}
where $\ket{\psi_{0}}$ is the 2d cluster state fixed-point wavefunction. From the phase diagram of the dual model, we can conclude that when $\lambda < 1$, $\ket{\psi (\lambda)}$ is a wavefunction in the $Z_2^{(0)} \times Z_2^{(1)}$ SPT phase and when $\lambda > 1$, $\ket{\psi (\lambda)}$ is a wavefunction in the trivial phase.

\subsection{Strange Correlator in the presence of noise}
In the presence of Pauli-Z decoherence on edge qubits, we want to study how much quantum information can be transmitted between two vertex qubits following similar procedure as in Sec~.\ref{sec: connection to type I SC}. \newline
\indent As was discussed in Sec.~\ref{sec: connection to type I SC}, we want to compute the trajectory averaged mutual information $I_m(L:R)$, which simplifies to the calculation of type-\RN{1} strange correlators when different trivial density matrices $\rho_m$ are inserted.\newline
\indent The denominator of $\text{SC}^{\RN{1}}$ can be computed as
\begin{equation}
    \begin{split}
        &\tr[\mathcal{E}[\rho] \rho_m] \\
        &= \tr[\rho \mathcal{E}^*[\rho_m]] \\
        &\propto \bra{\psi_0} e^{\lambda \sum_{e} \hat{X}_e} \prod_v \frac{1 + m_v \hat{X}_v}{2} \prod_{e} \frac{1 + (1 - 2p)m_e \hat{X}_e}{2} \ket{\psi_0},
    \end{split}
\end{equation}
the only terms in the product $\prod_v \frac{1 + m_v \hat{X}_v}{2}$ that would contribute are $\frac{1}{2^{N_v}}$ and $\frac{1}{2^{N_v}} \prod_v \hat{X}_v$, which after taking the expectation value gives $\frac{1}{2^{N_v - 1}}$. The rest of the terms can further be computed as
\begin{equation}
    \begin{split}
        &\tr[\mathcal{E}[\rho] \rho_m] \\
        &\propto \bra{\psi_0} \prod_e \left[(\cosh(\lambda) + \sinh(\lambda)\hat{X}_e) \frac{1 + (1 - 2p)m_e \hat{X}_e}{2}\right] \ket{\psi_0} \\
        &\propto \bra{\psi_0} \prod_e \Big[1 + \tanh(\lambda)(1- 2p)m_e \\
        &~~~~~~~~~~~~~~+ ({\tanh(\lambda) + (1 - 2p)m_e})\hat{X}_e \Big] \ket{\psi_0} \\
        &\propto \sum_{\gamma} \left(\prod_{e \in \gamma} [\tanh(\lambda) + (1-2p)m_e]\hat{X}_e\right) \\
        &~~~~~~~~~~~~~~\left( \prod_{e \notin \gamma}[1 + \tanh(\lambda)(1-2p)m_e]\right),
    \end{split}
\end{equation}
where to get the last line, we notice that in the product over all edges $e$, only edges that form a closed loop $\gamma$ on the dual lattice would contribute. Moreover, we can represent the summation over all possible loops configurations on the dual lattice as a partition function of an Ising model, where the Ising spins $\sigma_i$ are placed on the center of each plaquette,
\begin{equation}\label{eq: stat model of perturbed 2d cluster state}
    \begin{split}
        \tr[\mathcal{E}[\rho_0]\rho_m] \propto \sum_{\{\sigma\}} e^{\sum_{<ij>} \beta J_{ij}\sigma_i \sigma_j},
    \end{split}
\end{equation}
where for clarity we relabel each edge $e$ by $<ij>$ which denotes the bond that connects the two nearest-neighbor sites $i$ and $j$ on the dual lattice. The Ising coupling is given by
\begin{equation} \label{eq: coupling const of perturbed 2d cluster state}
    \beta J_{ij} = \tanh^{-1} \left[ \frac{\tanh(\lambda) + (1 - 2p)m_{ij}}{1 + \tanh(\lambda)(1 - 2p)m_{ij}} \right].
\end{equation}

Employing the same techniques, we can compute the numerator of $\text{SC}^{\RN{1}}$,
\begin{equation}
    \begin{split}
        &\tr[\mathcal{E}[\rho] Z_{v_i} Z_{v_j} \rho_m]\\
        &\propto \bra{\psi_0} e^{\lambda \sum_{e} \frac{\hat{X}_e}{2}} Z_{v_i} Z_{v_j} \prod_{v} \frac{1 + m_v X_v}{2} \prod_{e} \frac{1 + (1 - 2p)m_e X_e}{2} \\
        &~~~~~~~~~~~~~~~~~~~~~~~~~~~~~~~~~~~~~~~~~~~~~~~~~~e^{\lambda \sum_{e} \frac{\hat{X}_e}{2}} \ket{\psi_0} \\
        &\propto \bra{\psi_0} e^{\lambda \sum_{e} \frac{\hat{X}_e}{2}} \prod_{e \in l} X_e \prod_{v} \frac{1 + m_v X_v}{2} \prod_{e} \frac{1 + (1 - 2p)m_e X_e}{2} \\
        &~~~~~~~~~~~~~~~~~~~~~~~~~~~~~~~~~~~~~~~~~~~~~~~~~~e^{\lambda \sum_{e} \frac{\hat{X}_e}{2}} \ket{\psi_0} \\
        &\propto \bra{\psi_0} e^{\lambda \sum_{e} \hat{X}_e} \prod_{e \in l} X_e \prod_{e} \frac{1 + (1 - 2p)m_e X_e}{2} \ket{\psi_0} \\
        &\propto \bra{\psi_0} \prod_{e \in l} \hat{X}_e \prod_e \Big[1 + \tanh(\lambda)(1- 2p)m_e \\
        &+ ({\tanh(\lambda) + (1 - 2p)m_e})\hat{X}_e\Big] \ket{\psi_0}
    \end{split}
\end{equation}
where we have used the fact $Z_{v_i}Z_{v_j}$ when hit on $\bra{\psi_0}$ gives $\prod_{e \in l}X_e$ and $l$ is a string on the direct lattice connecting $v_i$ and $v_j$. Since only closed loop of $\hat{X}_e$ has non-zero expectation value, the numerator of $\text{SC}^{\RN{1}}$ can be written as
\begin{equation}
\begin{split}
    &\tr[\mathcal{E}[\rho] Z_{v_i} Z_{v_j} \rho_m]\\
    &\propto \sum_{\gamma'} \left(\prod_{e \in \gamma'} [\tanh(\lambda) + (1-2p)m_e]\hat{X}_e\right) \\
    &\left( \prod_{e \notin \gamma'}[1 + \tanh(\lambda)(1-2p)m_e]\right) \\
    &\propto \left\langle \sigma_{v_i} \sigma_{v_j} \right\rangle_{\beta J}
\end{split}
\end{equation}
where $\gamma'$ is an open string whose end points are $v_i$ and $v_j$. Therefore we can see the numerator is nothing but the $\sigma_{v_i}\sigma_{v_j}$ correlation function of the stat mech model we find in Eq.\eqref{eq: stat model of perturbed 2d cluster state}. \newline

\subsection{Behavior of $I_c(L:RM)$ at small $\lambda$}

\indent When $ 0 <\lambda \ll 1$, the coupling constant in Eq.\eqref{eq: coupling const of perturbed 2d cluster state} is approximatly $\tanh^{-1}(1- 2p) m_{ij} + \lambda + O(\lambda^2)$.
Therefore, we can identify the statistical mechanics model in Eq.\eqref{eq: stat model of perturbed 2d cluster state} governed by a perturbed RBIM Hamiltonian
\begin{equation}
    H[m, \lambda] = -\sum_{\left\langle i,j \right\rangle} \left( m_{ij} \sigma_i \sigma_j + \frac{\lambda}{\tanh^{-1}(1-2p)} \sigma_i \sigma_j \right)
\end{equation}
at inverse temperature $\beta = \tanh^{-1}(1-2p)$. Such a Hamiltonian describes a disorder ensemble where a bond configuration $m$ occur with probability \begin{equation} \label{eq: perturbed partition function}
    \begin{split}
        p\left(m, \lambda \right) \propto Z\left( m, \lambda\right) \equiv \sum_{\{\sigma\}} e^{-\beta H[m, \lambda]}.
    \end{split}
\end{equation} \newline
\indent Although we are not able to solve the above model exactly, we notice that the perturbation $\frac{\lambda}{\tanh^{-1}(1-2p)} \sigma_i \sigma_j$ favors the ferromagnetic phase. Compared to the partition function of RBIM along the Nishimori line in \eqref{eq: RBIM partition function} (which is also the $\lambda = 0$ limit in \eqref{eq: stat model of perturbed 2d cluster state}), at the same inverse temperature $\beta$, $Z\left(m, \lambda \right) > Z\left( m, \lambda = 0\right)$. So we expect at low temperature, bond configurations $m$ with ferromagnetic/long-ranged order (as diagonoised by $\abs{\left\langle \sigma_{v_i} \sigma_{v_j} \right\rangle_{\beta, m, \lambda}}$) are weighted more as we turn on the interaction. Since the RBIM along the Nishimori line can stabilize a ferromagnetic phase at low temperature, we also expect the disorder ensemble described by \eqref{eq: perturbed partition function} can stabilizer a ferromagnetic phase. While at high enough temperature, the ensemble enters the paramagnetic phase. Similar to what we have argued in Sec~.\ref{sec: connection to type I SC}, averaging the mutual information $I^{(m)}(L:R)$ with probability $p(m,\lambda)$ is greater than $1$ in the ferromagnetic phase and equal to $1$ in the paramagnetic phase. Therefore, the phase diagram outlined by $I_{c}(L:RM)$ is similar to the one we find in Sec~.\ref{sec: connection to type I SC}, where it is greater than $0$ at small decoherence strength and equals to $0$ at large decoherence strength. The phase transition that separates the two phases is likely to be in the RBIM universality class.

\section{Mixed SPT as quantum channels}\label{appendix: quantum channel}
\subsection{1d Cluster state}
We start with a cluster state of size $2N+1$ where the odd and even sublattice both have $Z_2$ symmetry generated by $\prod_{k=0}^{N-1} X_{2k}$ and $\prod_{k=0}^{N-1} X_{2k+1}$. We label the even and odd sublattice as $A$ and $B$ respectively. We denote the qubit at $0$ as the left boundary $L$ and the qubit at $2N$ as the right boundary $R$ so that both boundaries lie in the $A$ sub-lattice. We measure all qubits except the $L,R$ in the $X$ basis. Without decoherence, it is known that measuring these qubits implements the following channel on the $L$ qubit,
\begin{align}
    \rho_L(T) = \sum_{\{m\}} \cdots &X^{m_3} Z^{m_2} X^{m_1} \rho_L(0) X^{m_1} Z^{m_2} X^{m_3} \cdots \otimes \nonumber \\  &\ket{m_1,m_2,m_3,\dots}\bra{m_1,m_2,m_3,\dots},
\end{align}
where $T=2N-1$ is the virtual time for which qubit at $L$ is evolved; $m_1,m_2,\dots $ are the measurement outcomes of the bulk measurement. The evolution on $L$ is unitary if all measurement outcomes are known. In other words, the information about $L$ is transferred in the virtual time direction \textit{unitarily}.

Now we decohere the $B$ sub-lattice, which are odd qubits. The decoherence is assumed to act in the following way\begin{align*}
    \mathcal{N}_Z[\rho] = (1-p)\rho + pZ\rho Z.
\end{align*}
The qubits on the $B$ sub-lattice are measured after the decoherence has acted. Since \begin{align*}
    \mathcal{N}_i \left[\frac{1\pm X}{2}\rho \frac{1\pm X}{2}\right] = 
    &(1-p)\frac{1\pm X}{2}\rho \frac{1\pm X}{2} + \\
    + &p\frac{1\mp X}{2}\rho \frac{1+\mp X}{2},
\end{align*}
the decoherence flips with the measurement outcome with probability $p$. The channel on $L$ in the presence of the decoherence is now given by,
\begin{align}
    \rho_L(T) = \sum_{\{m\}}\rho_{\{m\},L}  \otimes \ket{m_1,m_2,m_3,\dots}\bra{m_1,m_2,m_3,\dots},
\end{align}
with \begin{align}
    \rho_{\{m\},L}  & = \prod_{t\in \mathrm{odd}} \mathcal{N}_{X} \left[ \left( Z^{m_{2t}} X^{m_{2t-1}} \right)^\dagger \rho_L(0)  X^{m_{2t-1}}Z^{m_{2t}} \right]
\end{align}
See Fig.~\ref{fig:spt as channel}. For a given measurement outcome $\{m\}$ the evolution of $L$ also has $X$ decoherence with strength $p$.The amount of information stored in a qubit transmitted in the presence of decoherence at late times is given approximately by $(1-2p)^{2N}/(2\ln 2)$. Thus $I_c(L:RM)=(1-2p)^{2N}/(2\ln 2)$. This is an independent calculation of $I_c(L:RM)$ and matches with the calculation done using eq.\eqref{eq: Ic(RM) Z dephaing one sublattice}.

\begin{figure}[bt!]
    \centering
    \includegraphics[width=\linewidth]{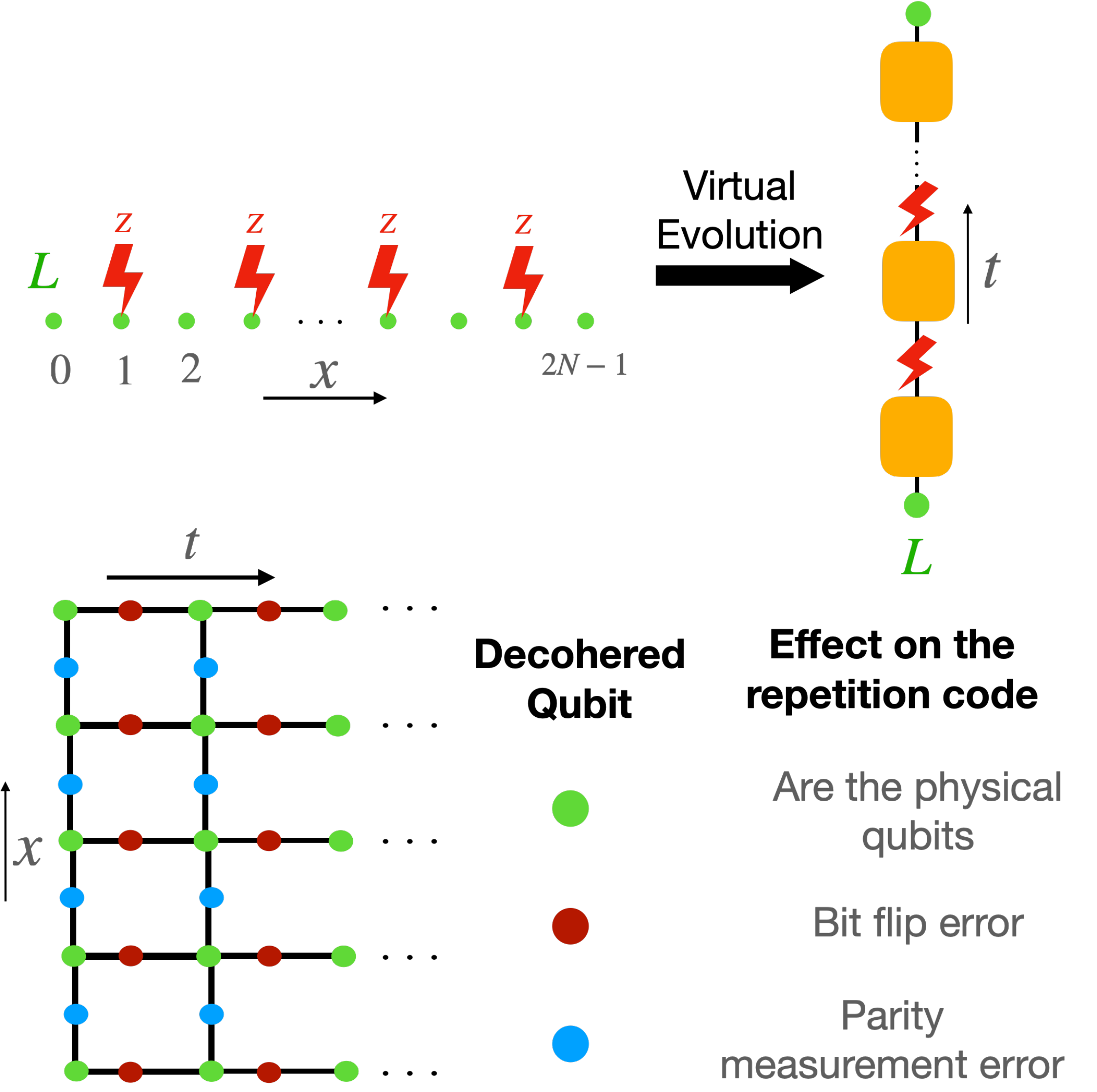}
    \caption{SPTs for pure states are known to be resource states for measurement-based quantum computing (MBQC). Decoherence in the SPT leads to noisy evolution of MBQC in the virtual time direction. \textit{Top.} Noise on one sub-lattice of the 1d cluster state leads to $X-$error on the virtual time evolution. \textit{Bottom.} 2d cluster state is a resource for transmission of 1d repetition code post-measurement. Decoherence on the edge qubits leads to bit-flip and measurement errors in the virtual dynamics of the repetition code. Error threshold in the virtual dynamics is related to the transitions for mixed-state SPT.}
    \label{fig:spt as channel}
\end{figure}

\subsection{2d Cluster state}
The 2-dimensional cluster state is a foliated version of the repetition code~\cite{Bolt_2016}. Let the qubits at the vertices of one boundary be the $L$ boundary and those at the opposite side be the $R$ boundary. The measurements of the edge qubits and the bulk vertex qubits implement a coherent evolution of the boundary. Fig.~\ref{fig:spt as channel} shows a layer of the cluster state. The green qubits at the left boundary are encoding the repetition code. Measuring blue qubits amounts to performing parity measurement $Z_i Z_{i+1}$ in the repetition code. The red qubits are measured in X basis and if the measurement outcome is $-1$ then a bit flip occurs at the corresponding green qubit. 

Decoherence on links (red and blue qubits) introduce errors in the measurement outcomes of these qubits, similar to the 1d case. The measurement error on the blue qubits will lead to a measurement error in the parity check and error on the outcome of red qubits results in X-dephasing of the repetition code. Thus the post-measurement state of the cluster state can be thought of as dynamics of the repetition code under X-dephasing and parity measurement error with equal strength $p$. This dynamic is known to have an error threshold below which the logical space of the repetition code is protected against X-dephasing and gets destroyed above the threshold. The quantum error correction transition is known to be described by the random bond Ising model (RBIM).

\end{document}